\documentclass[12pt]{article}
\usepackage{amsmath}
\usepackage{graphicx,psfrag,epsf}
\usepackage{enumerate}
\usepackage{natbib}
\usepackage{textcomp}
\usepackage[hyphens]{url} 
\usepackage[bookmarks=false]{hyperref}

\newcommand{\blind}{0}

\providecommand{\tightlist}{%
  \setlength{\itemsep}{0pt}\setlength{\parskip}{0pt}}

\usepackage{longtable,booktabs,array}
\usepackage{calc} 
\usepackage{etoolbox}
\makeatletter
\patchcmd\longtable{\par}{\if@noskipsec\mbox{}\fi\par}{}{}
\makeatother
\IfFileExists{footnotehyper.sty}{\usepackage{footnotehyper}}{\usepackage{footnote}}
\makesavenoteenv{longtable}

\usepackage{xcolor}
\usepackage{amsfonts}
\usepackage{xr-hyper}
\usepackage{booktabs}
\usepackage{longtable}
\usepackage{array}
\usepackage{multirow}
\usepackage{wrapfig}
\usepackage{float}
\usepackage{colortbl}
\usepackage{pdflscape}
\usepackage{tabu}
\usepackage{threeparttable}
\usepackage{threeparttablex}
\usepackage[normalem]{ulem}
\usepackage{makecell}
\usepackage{xcolor}

\usepackage{amsthm}
\newtheorem{theorem}{Theorem}[section]

\theoremstyle{definition}
\newtheorem{definition}{Definition}[section]
\theoremstyle{definition}

\theoremstyle{definition}

\theoremstyle{definition}

\theoremstyle{remark}

\begin{document}

\def\spacingset#1{\renewcommand{\baselinestretch}%
{#1}\small\normalsize} \spacingset{1}


\if0\blind
{
  \title{\bf Fast Bayesian Record Linkage for Streaming Data Contexts}

  \author{
        Ian Taylor \\
    Department of Statistics, Colorado State University\\
     and \\     Andee Kaplan \\
    Department of Statistics, Colorado State University\\
     and \\     Brenda Betancourt \\
    NORC at the University of Chicago\\
      }
  \maketitle
} \fi

\if1\blind
{
  \bigskip
  \bigskip
  \bigskip
  \begin{center}
    {\LARGE\bf Fast Bayesian Record Linkage for Streaming Data Contexts}
  \end{center}
  \medskip
} \fi

\bigskip
\begin{abstract}
Record linkage is the task of combining records from multiple files which refer to overlapping sets of entities when there is no unique identifying field. In streaming record linkage, files arrive sequentially in time and estimates of links are updated after the arrival of each file. This problem arises in settings such as longitudinal surveys, electronic health records, and online events databases, among others. The challenge in streaming record linkage is to efficiently update parameter estimates as new data arrive. We approach the problem from a Bayesian perspective with estimates calculated from posterior samples of parameters and present methods for updating link estimates after the arrival of a new file that are faster than fitting a joint model with each new data file. In this paper, we generalize a two-file Bayesian Fellegi-Sunter model to the multi-file case and propose two methods to perform streaming updates. We examine the effect of prior distribution on the resulting linkage accuracy as well as the computational trade-offs between the methods when compared to a Gibbs sampler through simulated and real-world survey panel data. We achieve near-equivalent posterior inference at a small fraction of the compute time. Supplemental materials for this article are available online.
\end{abstract}

\noindent%

\vfill

\newpage
\spacingset{1.5} 

\hypertarget{introduction}{%
\section{Introduction}\label{introduction}}

Record linkage is the task of resolving duplicates in two or more overlapping sets of records, or files, from multiple noisy data sources, often without the benefit of having a unique identifier. For example, in a longitudinal survey setting it is possible to have multiple responses from the same person with misspellings or other data errors. This type of error is shown in Table \ref{tab:example-data}, where records 1 and 5 represent responses from the same person that were stored with a misspelling in the given name. This presents a problem for those that wish to use this data to make inferences. With the current accessibility and continuity of data, record linkage has become crucial for many areas of application including healthcare \citep{fleming2012record, HofRavelliZwinderman17}, official statistics \citep{winkler_2006, kaplan2022practical, wortman2019record}, and fraud detection and national security \citep{Vatsalan2017privacy}.

Although probabilistic approaches for record linkage have become more common in recent years, principled approaches that are computationally tractable and scalable for large data sets are limited \citep{binette2021some}. Moreover, existing approaches are not suited for streaming data settings, where inference is desired continuously. In the streaming context, data files are expected to arrive sequentially in time with no predetermined number of files. A limited portion of the machine learning literature has targeted the area of near real-time record linkage from a data-driven perspective \citep{Christen09, Ioannou2010OntheFlyEQ, Dey11online, Altwaijry17, Karapiperis2018}.

In this work, we propose new methodology to perform record linkage with streaming data in an efficient and statistically principled fashion under a Bayesian framework. A model-based approach, such as the one we propose, provides interpretable parameters and a way to encode prior knowledge about the data generation process. Bayesian inference also provides natural uncertainty quantification, allowing uncertainty from record linkage to propagate to downstream analysis \citep{kaplan2022practical}. This work presents the first model-based approach to perform record linkage in streaming data contexts.

\begin{table}

\caption{\label{tab:example-data}An example of noisy data in need of deduplication. Rows 1 and 5 refer to the same entity but differ due to an error in `Given Name'.}
\centering
\begin{tabular}[t]{lllr}
\toprule
Given Name & Surname & Age & Occupation\\
\midrule
\cellcolor[HTML]{DFDFDF}{maddisom} & \cellcolor[HTML]{DFDFDF}{ryan} & \cellcolor[HTML]{DFDFDF}{f} & \cellcolor[HTML]{DFDFDF}{3}\\
marleikh & hoffman & d & 4\\
samara & pater5on & d & 5\\
lili & wheatlry & f & 7\\
\cellcolor[HTML]{DFDFDF}{maddison} & \cellcolor[HTML]{DFDFDF}{ryan} & \cellcolor[HTML]{DFDFDF}{f} & \cellcolor[HTML]{DFDFDF}{3}\\
\bottomrule
\end{tabular}
\end{table}

A significant portion of the probabilistic record linkage literature has focused on linking two data files \citep{fellegi1969theory, tancredi2011hierarchical, gutman_2013, sadinle2017bayesian}. Recently, Bayesian approaches for multi-file record linkage have become popular \citep{sadinle2013generalized, sadinle2014detecting, steorts2016bayesian, betancourt2016flexible, aleshinguendel2021multifile}. In particular, \citet{aleshinguendel2021multifile} extend the Bayesian Fellegi-Sunter model of \citet{sadinle2017bayesian} through the use of a partition prior. However, the existing literature is limited to non-streaming settings where the number of files is fixed and known in advance, and record linkage is performed offline in a single procedure. Recent advances have made record linkage possible for big offline data settings, either by jointly performing blocking and entity resolution \citep{marchant2021dblink} or by quickly computing point estimates and approximating the posterior distribution \citep{mcveigh2020scaling}. Nonetheless, these approaches are not suited to efficiently assimilate new data. To address this gap in the literature from a fully model-driven perspective, we focus on developing a Bayesian model for multi-file record linkage that enables online data scenarios. Our approach uses recursive Bayesian computation techniques to produce samples from the full posterior that efficiently update existing draws from the previous posterior. To date, such recursive Bayesian updates have not been used for linkage in a streaming setting. Our proposed model is constructed under the Fellegi-Sunter paradigm, which entails pairwise comparisons of records \citep{fellegi1969theory, sadinle2013generalized}. We explore diffuse and informative prior distributions and provide two streaming samplers.

The remainder of this paper proceeds as follows. Section \ref{sec:bayesian-record-linkage-model} defines the Bayesian record linkage model for streaming data and defines the problem context, notation, assumptions, and constraints for the model. Section \ref{sec:streaming-sampling} introduces two streaming samplers which can be used to perform updates of parameter estimates upon the arrival of a new file. Section \ref{sec:simulation-study} evaluates these methods on both the quality of samples they produce as well as their speed on simulated data sets. Section \ref{sec:poland} provides the result of performing streaming record linkage on real-world survey panel data. Section \ref{sec:discussion} contains discussion of further advantages and disadvantages of each streaming update method.

\hypertarget{sec:bayesian-record-linkage-model}{%
\section{Bayesian Record Linkage Model for Streaming Data}\label{sec:bayesian-record-linkage-model}}

We will begin this section with a description of the streaming data context, definition of notation, and enumeration of assumptions. We then define the likelihood and prior specification for the multi-file record linkage model.

\hypertarget{streaming-record-linkage-notation}{%
\subsection{Streaming Record Linkage Notation}\label{streaming-record-linkage-notation}}

We consider \(k\) files \(X_1, \dots, X_k\) that are collected temporally, so that file \(X_m\) is available at time \(T_m\), with \(T_1 < T_2 < \dots < T_k\). See Figure \ref{fig:streaming-files} in Appendix \ref{sec:supplemental-figures} for a diagram depicting this context. Each file \(X_m\) contains \(n_m \geq 1\) records \(X_m = \{\boldsymbol x_{mi}\}_{i = 1}^{n_m}\), with each \(n_m\) potentially distinct. Each record is comprised of \(p_m\) fields and it is assumed that there is a common set of \(F\) fields numbered \(f=1,\dots,F\) across the \(k\) files which can be numeric, text or categorical. Records representing an individual (or entity) can be noisily duplicated across files. Each individual or entity is recorded at most once in each file, corresponding to an assumption that there are no duplicates within a file. This setting has a growing complexity--- with \(k\) files, all records in \(k(k-1)/2\) pairs of files must be compared and linked. The goal of the record linkage problem is identifying which records in files \(X_1, \dots, X_k\) refer to the same entities. This context is considered ``streaming'' because data is continuously generated with no predetermined stopping point and our goal is to update the inference pipeline as new information becomes available.

Our record linkage model for the streaming data context extends the ideas of \citet{fellegi1969theory} and \citet{sadinle2017bayesian}. Within this paradigm, the comparisons are assumed to come from one of two distributions, \(\cal M\) for coreferent pairs and \(\cal U\) for non-coreferent pairs. Two records are coreferent if they refer to the same entity. The \citet{fellegi1969theory} framework was extended to the Bayesian paradigm for two-file record linkage in \citet{sadinle2017bayesian}. In this work, we further extend the model for a general \(k\)-file scenario. In contrast to the \citet{aleshinguendel2021multifile} model which also extends the \citet{sadinle2017bayesian} model for the multi-file case, we parameterize the record matching as vectors linking to the most recent previous occurrence of an individual and place an informative prior on these vectors to avoid overlinking between files. This parameterization is the mechanism by which streaming updates are possible.

We denote comparison between two records, \(\boldsymbol x_{m_1 i}\) in file \(X_{m_1}\) and \(\boldsymbol x_{m_2 j}\) in file \(X_{m_2}\), as a function, \(\gamma(\boldsymbol x_{m_1 i}, \boldsymbol x_{m_2 j})\), which compares the values in each field, \(f\), dependent on field type. Each comparison results in discrete levels \(0,\dots, L_f\) with 0 representing exact equality and subsequent levels representing increased difference. For example, categorical values can be compared in a binary fashion, numerical fields can be compared by binned absolute difference, and text fields can be compared by binned Levenshtein distance \citep{christen2012data}. We define \(P = \sum_{f=1}^F (L_f+1)\), as the total number of levels of disagreement of all fields. The comparison \(\gamma(\boldsymbol x_{m_1 i}, \boldsymbol x_{m_2 j})\) takes the form of a \(P\)-vector of binary indicators containing \(F\) ones and \(P-F\) zeros which indicates the level of disagreement between \(\boldsymbol x_{m_1 i}\) and \(\boldsymbol x_{m_2 j}\) in each field. Exactly one 1 must appear in the first \(L_1+1\) elements of \(\gamma(\boldsymbol x_{m_1 i}, \boldsymbol x_{m_2 j})\), one 1 in the next \(L_2+1\) elements, and so on. The comparison vectors are collected into matrices \(\Gamma^{(1)}, \dots, \Gamma^{(k-1)}\) where \(\Gamma^{(m-1)}\) contains all comparisons between the records in file \(X_m\) and previous files. The comparison matrix \(\Gamma^{(m-1)}\) has \(n_m\cdot (n_1 + \dots + n_{m-1})\) rows and \(P\) columns. Define \(\Gamma^{(1:m)}\) as \(\{\Gamma^{(1)}, \dots, \Gamma^{(m)} \}\) for \(m \in 1, \dots, k - 1\).

Records can be represented as a \(k\)-partite graph, with nodes representing records in each file and a link between two records indicating that they are coreferent. This graph can be segmented according to the order of files. First, a bipartite graph between \(X_1\) and \(X_2\); then a tripartite graph between \(X_1, X_2\), and \(X_3\), where records in \(X_3\) link to records in \(X_1\) and \(X_2\); until finally a \(k\)-partite graph between \(X_1,\dots,X_k\) where records in \(X_k\) link to records in \(X_1,\dots,X_{k-1}\). These graphs can be represented with \(k-1\) matching vectors, with one vector per file \(X_2, \dots, X_k\). Each vector, denoted \(\boldsymbol Z^{(m-1)}\), has length \(n_m\) with the value in index \(j\), denoted \(Z^{(m-1)}_j\), corresponding to the record \(\boldsymbol x_{mj}\) as follows,
\[
Z^{(m-1)}_j = \begin{cases}
\sum_{\ell=1}^{t-1} n_\ell + i & \parbox[t]{.60\textwidth}{for $t < m$, if $\boldsymbol x_{ti} \in X_t$ and $\boldsymbol x_{mj}$ are coreferent,} \\ 
\sum_{\ell=1}^{m-1} n_\ell +j & \text{otherwise.}
\end{cases}
\]

Let \(\boldsymbol Z^{(m-1)} = \left(Z^{(m-1)}_j\right)_{j=1}^{n_{m}}\) and \(\boldsymbol Z^{(1:m)} = \left\{\boldsymbol Z^{(1)},\dots,\boldsymbol Z^{(m)}\right\}\) for \(m \in 1, \dots, k - 1\). These vectors identify which records are coreferent and are therefore the main parameters of interest in the record linkage problem.

We also define parameters \(\boldsymbol m\) and \(\boldsymbol u\), which specify the distributions \(\cal M\) and \(\cal U\) respectively. Both \(\boldsymbol m\) and \(\boldsymbol u\) are \(P\)-vectors which can be separated into the sub-vectors \(\boldsymbol m = \begin{bmatrix}\boldsymbol m_1 & \dots & \boldsymbol m_F\end{bmatrix}\) and \(\boldsymbol u = \begin{bmatrix}\boldsymbol u_1 & \dots & \boldsymbol u_F\end{bmatrix}\), where \(\boldsymbol m_f\) and \(\boldsymbol u_f\) have length \(L_f+1\). Then \({\cal M}(\boldsymbol m) = \prod_{f=1}^F \text{Multinomial}(1; \boldsymbol m_f)\) and \({\cal U}(\boldsymbol u) = \prod_{f=1}^F \text{Multinomial}(1; \boldsymbol u_f)\) are the distributions for matches and non-matches, respectively.

\hypertarget{sec:preserving-transitivity}{%
\subsection{Preserving the Duplicate-Free File Assumption}\label{sec:preserving-transitivity}}

Preserving the assumption of duplicate-free files with a large number of files is a challenge because the combination of several links throughout the parameters \(\boldsymbol Z^{(1:(k-1))}\) may imply that two records in the same file are coreferent. For example if \(Z^{(1)}_1 = 1\), \(Z^{(2)}_1 = 1\), and \(Z^{(2)}_2 = n_1 + 1\), then the records \(\boldsymbol x_{31}\) and \(\boldsymbol x_{32}\) are implied to be coreferent even though they are not directly linked to the same record. We address this by placing constraints on the values of these parameters such that no two records may link directly to the same record in a previous file. Because each record can send at most one link to a previous record and receive at most one link from a later record, we guarantee that no two records in the same file are transitively linked. Figure \ref{fig:valid-links} depicts a three-file example of prohibited and allowed values of \(\boldsymbol Z^{(1)}\) and \(\boldsymbol Z^{(2)}\). Both values are logically equivalent, but without this constraint the prohibited configuration could allow for one record in file \(X_4\) to link to record \(\boldsymbol x_{31}\) while another links to record \(\boldsymbol x_{21}\), becoming coreferent and violating the assumption.

The bipartite matching, \(\boldsymbol Z^{(1)}\), is constrained in a manner consistent with \citet{sadinle2017bayesian}. Namely, that there can be no two \(Z^{(1)}_i = Z^{(1)}_{i'}\) where \(i \neq i'\). The tripartite matching must be similarly restricted to enforce our link validity constraint. Specifically, for some \(1 \leq i \leq n_3\) and \(1 \leq j \leq n_1\), \(Z^{(2)}_i\) cannot equal \(j\) if \(Z^{(1)}_{k} = j\) for any \(k \leq n_2\). That is, record \(i\) cannot be linked to a record \(j\) in \(X_1\) which already has a match in \(X_2\). To enforce transitivity of the coreference relationship, comparisons with files \(X_m, m \geq 3\) will be constrained.

\begin{definition}
\protect\hypertarget{def:link-validity}{}\label{def:link-validity}\textbf{Link Validity Constraint.} Let \(\mathcal{C}_k\) be the set of all matching vectors \(\boldsymbol Z^{(1:(k-1))}\) such that every record \(\boldsymbol x_{m_1 i}\) receives at most one link from a record \(\boldsymbol x_{m_2 j}\) where \(m_2 > m_1\). That is, there is at most one value in any \(\boldsymbol Z^{(m_2-1)}\) with \(m_2 > m_1\) that equals \(\sum_{\ell=1}^{m_1 - 1} n_\ell + i\). Matching vectors \(\boldsymbol Z^{(1:(k-1))}\) are valid if and only if \(\boldsymbol Z^{(1:(k-1))} \in \mathcal{C}_k\).
\end{definition}

\begin{figure}

{\centering \includegraphics[width=0.45\linewidth]{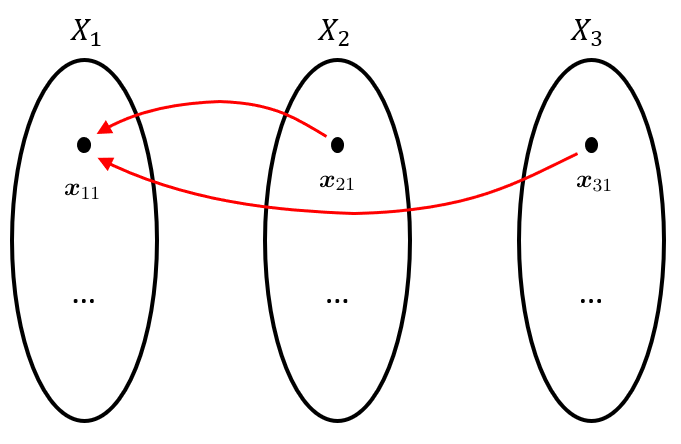} \includegraphics[width=0.45\linewidth]{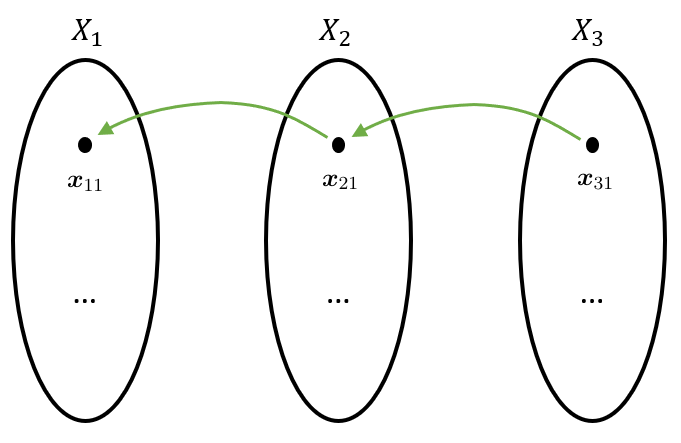} 

}

\caption{Examples of both prohibited (left) and allowed (right) links between records in three files. On the left $Z^{(1)}_1 = 1$ and $Z^{(2)}_1 = 1$, while on the right $Z^{(1)}_1 = 1$ and $Z^{(2)}_1 = n_1 + 1$. The left configuration is prohibited because the record $\boldsymbol x_{11}$ receives a link from both $\boldsymbol x_{21}$ and $\boldsymbol x_{31}$. Both configurations define the same  cluster containing these three records.}\label{fig:valid-links}
\end{figure}

This constraint aids in the identifiability of the parameters \(\boldsymbol Z^{(1:(k-1))}\). Under these constraints each logical cluster of at most one record from each file has one unique valid representation, namely a chain of links from the latest-appearing record to the earliest-appearing record, linking records in order of appearance. The chain nature aids in computation--- it becomes possible to list all members of a cluster by starting at one of its members and traversing the chain forwards and backwards without needing to branch or double back.

\hypertarget{sec:streaming-likelihood}{%
\subsection{Likelihood}\label{sec:streaming-likelihood}}

In this section and Section \ref{sec:prior-specification}, we define the likelihood and priors that contribute to the streaming record linkage model posterior. The full posterior distribution is presented in Appendix \ref{sec:posterior-distribution}. Consistent with the formulation in \citet{sadinle2017bayesian}, the likelihood for the two-file case is defined as
\[P(\Gamma^{(1)} | \boldsymbol Z^{(1)}, \boldsymbol m, \boldsymbol u) = \prod_{i=1}^{n_1} \prod_{j=1}^{n_2} P(\gamma_{ij} | \boldsymbol Z^{(1)}, \boldsymbol m, \boldsymbol u) = \prod_{i=1}^{n_1} \prod_{j=1}^{n_2} \prod_{f=1}^F \prod_{\ell=0}^{L_f} \left[m_{f\ell}^{\mathbb{I}(Z^{(1)}_j=i)} u_{f\ell}^{\mathbb{I}(Z^{(1)}_j\neq i)} \right]^{\gamma_{ij}^{f\ell}},\]
where \(\gamma_{ij} := \gamma(\boldsymbol x_{1i}, \boldsymbol x_{2j})\), \(\gamma_{ij}^{f\ell}\) is the component corresponding to level \(\ell\) of field \(f\), and \(\mathbb{I}(\cdot)\) is the indicator function taking a value of 1 if its argument holds and 0 otherwise. For every pair of records, one from each file, \(\boldsymbol m\) contributes to the distribution if the records are linked by \(\boldsymbol Z^{(1)}\) and \(\boldsymbol u\) contributes otherwise. We extend this to the \(k\)-file case by defining the match set,
\(M := M(\boldsymbol Z^{(1:(k-1))}) = \{(\boldsymbol x_{m_1 i}, \boldsymbol x_{m_2 j}): \boldsymbol x_{m_1 i}\text{ and }\boldsymbol x_{m_2 j}\text{ are linked}\},\)
to contain all pairs of records that are linked either directly or transitively through a combination of multiple vectors \(\boldsymbol Z^{(1:(k-1))}\). Testing whether \((\boldsymbol x_{m_1 i}, \boldsymbol x_{m_2 j}) \in M\) for \(m_1 < m_2\) is done by the process of \emph{link tracing}. This is the process by which we determine the links implied by transitivity in the match vectors. To perform link tracing, we start at \(\boldsymbol x_{m_2 j}\) and follow the values in \(\boldsymbol Z^{(1:(k-1))}\) to travel down the chain of links, starting with \(Z^{(m_2-1)}_j\). If \(\boldsymbol x_{m_1 j}\) is ever reached, then \((\boldsymbol x_{m_1 i}, \boldsymbol x_{m_2 j}) \in M\), while if a dead end is reached first, then \((\boldsymbol x_{m_1 i}, \boldsymbol x_{m_2 j}) \notin M\).

The full data model in the \(k\)-file case is then
\begin{equation}
P(\Gamma^{(1:(k-1))} | \boldsymbol m, \boldsymbol u, \boldsymbol Z^{(1:(k-1))}) = \prod_{m_1 < m_2}^k \prod_{i=1}^{n_{m_1}} \prod_{j=1}^{n_{m_2}} \prod_{f=1}^F \prod_{\ell=0}^{L_f} \left[m_{f\ell}^{\mathbb{I}((\boldsymbol x_{m_1 i}, \boldsymbol x_{m_2 j}) \in M)} u_{f\ell}^{\mathbb{I}((\boldsymbol x_{m_1 i}, \boldsymbol x_{m_2 j}) \notin M)} \right]^{\gamma^{f\ell}(\boldsymbol x_{m_1 i}, \boldsymbol x_{m_2 j})}. \label{eqn:streaming-data-model}
\end{equation}
The likelihood of the \(k\)-file Bayesian record linkage model encodes the assumption that all comparisons, \(\Gamma\), are conditionally independent given the parameters \(\boldsymbol m, \boldsymbol u, \boldsymbol Z^{(1:(k-1))}\). The same \(\boldsymbol m\) and \(\boldsymbol u\) probabilities appear in the distribution of comparisons between each pair of files, corresponding to an assumption of equal propensity for error in each file. Alternatively separate probabilities, \(\boldsymbol m_{t_1t_2}\) and \(\boldsymbol u_{t_1t_2}\), can be specified for the comparisons between files \(X_{t_1}\) and \(X_{t_2}\), as in \citet{aleshinguendel2021multifile}. However, every new file, \(X_k\), will require \(2(k-1)\) new parameters, \(\boldsymbol m_{t_1 k}\) and \(\boldsymbol u_{t_1 k}\) for all \(t_1 < k\), which may affect the model's computational performance as well as the ability of a sampler to adequately explore the space in a streaming setting. The support of the data distribution is dependent on the vectors \(\boldsymbol Z^{(1:(k-1))}\), specifically, the matching vectors must satisfy the link validity constraint given in Definition \ref{def:link-validity}. We explicitly write this constraint as an indicator function in the likelihood:
\begin{equation}
L(\boldsymbol m, \boldsymbol u, \boldsymbol Z^{(1:(k-1))}) = \mathbb{I}(\boldsymbol Z^{(1:(k-1))} \in \mathcal{C}_k)\cdot P(\Gamma^{(1:(k-1))} | \boldsymbol m, \boldsymbol u, \boldsymbol Z^{(1:(k-1))}). \label{eqn:streaming-likelihood}
\end{equation}

\hypertarget{sec:prior-specification}{%
\subsection{Prior Specification}\label{sec:prior-specification}}

\hypertarget{sec:m-u-prior}{%
\subsubsection{\texorpdfstring{Priors for \(\boldsymbol m\) and \(\boldsymbol u\)}{Priors for \textbackslash boldsymbol m and \textbackslash boldsymbol u}}\label{sec:m-u-prior}}

The parameters \(\boldsymbol m\) and \(\boldsymbol u\) are probabilities of a multinomial distribution, so we specify conjugate Dirichlet priors. Specifically, we let \(\boldsymbol m_f \sim \text{Dirichlet}(\boldsymbol a_f)\) and \(\boldsymbol u_f \sim \text{Dirichlet}(\boldsymbol b_f)\), for \(f=1,\dots,F\), where \(\boldsymbol a_f\) and \(\boldsymbol b_f\) are vectors with the same dimension, \(L_f+1\), as \(\boldsymbol m_f\) and \(\boldsymbol u_f\). For a diffuse prior we can set \(\boldsymbol a = \boldsymbol b = \boldsymbol 1\). Also it can be useful to encode prior knowledge about the propensity for duplicates to have errors in the prior for \(\boldsymbol m\). For example, if we know that an error in field \(f\) of a duplicated record has probability \(p\) of occurring, we can let
\begin{equation}
\boldsymbol a_f = s\cdot\begin{bmatrix}1-p & p/L_f & \dots & p/L_f\end{bmatrix}, \label{eqn:m-prior-construction}
\end{equation}
with \(s\) determining the strength of the prior knowledge. We empirically investigate the effect of this informative prior specification on \(\boldsymbol m\) in simulated data scenarios in Section \ref{sec:simulation-study}.

\hypertarget{sec:streaming-prior}{%
\subsubsection{\texorpdfstring{Priors for \(\boldsymbol Z^{(k-1)}\)}{Priors for \textbackslash boldsymbol Z\^{}\{(k-1)\}}}\label{sec:streaming-prior}}

We construct the prior for streaming matching vectors using the same hierarchy as specified in \citet{sadinle2017bayesian}. First, let \(w_j^{(k)} := \mathbb{I}\left(Z^{(k-1)}_j \leq \sum_{m=1}^{k-1} n_m\right)\), that is let \(w_j^{(k)}\) be an indicator that record \(j\) in file \(k\) is linked, and \(\boldsymbol w^{(k)} = \left\{w_j^{(k)}: j = 1, \dots, n_k\right\}\). Then to specify the prior for \(\boldsymbol Z^{(k-1)}\), we let
\begin{align}
\left. w_j^{(k)} \middle| \pi \right. &\stackrel{\text{iid}}{\sim} \text{Bernoulli}(\pi) \notag \\
\left. \boldsymbol Z^{(k-1)} \middle| \boldsymbol w^{(k)}\right. &\sim \text{Uniform}\left(\left\{\text{all valid $k$-partite matchings}\right\}\right). \label{eqn:z-prior}
\end{align}
Allowing \(\pi \sim \text{Beta}(\alpha_\pi, \beta_\pi)\) results in the marginal streaming prior
\begin{equation*}
P(\boldsymbol Z^{(k-1)}|\alpha_\pi,\beta_\pi) = \frac{(N - n_{k\cdot}(\boldsymbol Z^{(k-1)}))!}{N!}\cdot\frac{\mbox{B}(n_{k\cdot}(\boldsymbol Z^{(k-1)}) + \alpha_\pi, n_k - n_{k\cdot}(\boldsymbol Z^{(k-1)}) + \beta_\pi)}{\mbox{B}(\alpha_\pi, \beta_\pi)}, 
\end{equation*}
where \(N = \sum_{m=1}^{k-1} n_m\) and \(n_{k\cdot}(\boldsymbol Z^{(k-1)}) = \sum_{j=1}^{n_k} I(Z^{(k-1)}_j \leq N)\)

This streaming prior enforces the condition that no two records within the same file can link to the same record in a previous file. However, the more general link validity constraint in Definition \ref{def:link-validity} is not enforced in the prior. It is possible to enforce this constraint in the prior (rather than the likelihood) by specifying the prior as \(P(\boldsymbol Z^{(k-1)}|\alpha_\pi,\beta_\pi)\mathbb{I}(\boldsymbol Z^{(1:(k-1))} \in \mathcal{C}_k)\) with no effect on the posterior. By not enforcing the link validity constraint in Equation \ref{eqn:z-prior}, the resulting marginal streaming prior depends on \(N\), rather than on the number of records available to be linked based on \(\boldsymbol Z^{1:(k-2)}\). In empirical studies this decision has resulted in higher accuracy in the linkage. Further exploration of this prior specification is the subject of future research.

\hypertarget{sec:streaming-sampling}{%
\section{Streaming Sampling}\label{sec:streaming-sampling}}

The key to Bayesian streaming record linkage is an efficient means of updating the posterior distribution of existing parameters after the arrival of a new file, \(X_k\). In this section, we introduce two sampling approaches we have adapted to address this problem, Prior-Proposal-Recursive-Bayes (PPRB) and Sequential MCMC (SMCMC).

\hypertarget{prior-proposal-recursive-bayes-pprb}{%
\subsection{Prior-Proposal-Recursive Bayes (PPRB)}\label{prior-proposal-recursive-bayes-pprb}}

Prior-Proposal-Recursive Bayes is a recursive Bayesian sampling technique in which existing posterior samples from a previous stage are used as independent Metropolis proposals to sample from a later stage posterior distribution, conditioned on new data \citep{hooten2019making}. We consider a model with parameters \(\boldsymbol \theta\) and data \(\boldsymbol y_1, \boldsymbol y_2\):
\begin{equation*}
\boldsymbol y = \begin{bmatrix}\boldsymbol y_1 \\ \boldsymbol y_2 \end{bmatrix} \sim p(\boldsymbol y | \boldsymbol \theta) = p(\boldsymbol y_1 | \boldsymbol \theta)p(\boldsymbol y_2 | \boldsymbol \theta, \boldsymbol y_1),\ \ \boldsymbol \theta \sim p(\boldsymbol \theta)
\end{equation*}

We assume \(\boldsymbol y_1\) arrives before \(\boldsymbol y_2\) and posterior samples \(\boldsymbol \theta_{(1)}\dots \boldsymbol \theta_{(S)}\) are obtained from \(p(\boldsymbol \theta | \boldsymbol y_1)\). After \(\boldsymbol y_2\) arrives, these samples are resampled as independent Metropolis proposals for the updated posterior distribution \(p(\boldsymbol \theta | \boldsymbol y_1, \boldsymbol y_2)\). The acceptance ratio \(\alpha\) for the proposal \(\boldsymbol \theta'\) and current value \(\boldsymbol \theta\) simplifies to
\[\alpha = \min\left(\frac{p(\boldsymbol y_2 | \boldsymbol \theta', \boldsymbol y_1)}{p(\boldsymbol y_2 | \boldsymbol \theta, \boldsymbol y_1)},1\right).\]
This ratio depends only on the full conditional distribution of the new data, \(\boldsymbol y_2\) and so can be calculated quickly. If \(\boldsymbol y_2\) and \(\boldsymbol y_1\) are conditionally independent given \(\boldsymbol \theta\), then the old data \(\boldsymbol y_1\) does not need to be stored in order to calculate \(\alpha\) or perform PPRB.

To apply PPRB to the Bayesian record linkage model when a file \(X_k\) arrives, we have \(\boldsymbol y_2 = \Gamma^{(k-1)}\), \(\boldsymbol y_1 = \Gamma^{(1:(k-2))}\), and \(\boldsymbol \theta = \begin{bmatrix}\boldsymbol m & \boldsymbol u & \boldsymbol Z^{(1:(k-2))} \end{bmatrix}\). Since all comparisons are assumed conditionally independent given the parameters \(\boldsymbol m, \boldsymbol u, \boldsymbol Z^{(1:(k-2))}\), the past calculated comparisons \(\Gamma^{(1:(k-2))}\) would not be needed to calculate \(\alpha\) or perform PPRB. However, the streaming record linkage model requires additional parameters, \(\boldsymbol Z^{(k-1)}\), for the distribution of the new data, \(\Gamma^{(k-1)}\), so a straight forward application of PPRB is not possible. \citet{hooten2019making} propose drawing values of the new parameter from its predictive distribution and appending those values to the existing samples prior to PPRB, which retains the simplified form of the acceptance ratio, \(\alpha\). In the streaming record linkage problem, the predictive distribution of \(\boldsymbol Z^{(k-1)}\) reduces to its prior: \(p(\boldsymbol Z^{(k-1)} | \boldsymbol m, \boldsymbol u, \boldsymbol Z^{(1:(k-2))}, \Gamma^{(1:(k-2))}) = p(\boldsymbol Z^{(k-1)})\). However, because the space of possible values of \(\boldsymbol Z^{(k-1)}\) is on the order of \((\sum_{\ell=1}^{k-1} n_\ell)^{n_k}\) and the proposed prior is diffuse, these values are rarely good proposals for the updated posterior distribution, leading to low acceptance rates and slow mixing.

For this reason, we propose PPRB-within-Gibbs, a Gibbs sampler in which one of the steps is an independent Metropolis proposal from prior stage posterior samples.

\begin{definition}
\protect\hypertarget{def:pprb-within-gibbs}{}\label{def:pprb-within-gibbs}\textbf{PPRB-within-Gibbs algorithm}. Consider a general model with partitioned data \(\boldsymbol y_1, \boldsymbol y_2\), and parameters \(\boldsymbol \theta_1, \boldsymbol \theta_2, \boldsymbol \theta_3\):
\begin{align*}
\boldsymbol y_1 | \boldsymbol \theta_1, \boldsymbol \theta_2 &\sim p(\boldsymbol y_1 | \boldsymbol \theta_1, \boldsymbol \theta_2) \\
\boldsymbol y_2 | \boldsymbol \theta_1, \boldsymbol \theta_2, \boldsymbol \theta_3 &\sim p(\boldsymbol y_2 | \boldsymbol \theta_1, \boldsymbol \theta_2, \boldsymbol \theta_3) \\
\boldsymbol \theta_1 \sim p(\boldsymbol \theta_1),\ \boldsymbol \theta_2 &\sim p(\boldsymbol \theta_2),\ \boldsymbol \theta_3 \sim p(\boldsymbol \theta_3)
\end{align*}

The parameters \(\boldsymbol \theta_1, \boldsymbol \theta_2, \boldsymbol \theta_3\) have independent priors, \(\boldsymbol y_1, \boldsymbol y_2\) are conditionally independent given the parameters, and the first wave of data, \(\boldsymbol y_1\), is not dependent on \(\boldsymbol \theta_3\). Let there be existing posterior samples, \(\{\boldsymbol \theta_1^s\}_{s=1}^S\) from the distribution \(p(\boldsymbol \theta_1 | \boldsymbol y_1)\). Then for the desired number of posterior samples,

\begin{enumerate}
\def\labelenumi{\arabic{enumi}.}
\tightlist
\item
  Update the parameter \(\boldsymbol \theta_2\) from the full conditional distribution \([\boldsymbol \theta_2 | \boldsymbol \theta_1, \boldsymbol \theta_3, \boldsymbol y_1, \boldsymbol y_2]\),
\item
  (PPRB step) Propose a new value \(\boldsymbol \theta_1^\ast\) by drawing from the existing posterior samples \(\{\boldsymbol \theta_1^s\}_{s=1}^S\) with replacement. Accept or reject the proposal using the Metropolis-Hastings ratio \[\alpha = \min\left(\frac{p(\boldsymbol y_2 | \boldsymbol \theta_1^\ast, \boldsymbol \theta_2, \boldsymbol \theta_3)}{p(\boldsymbol y_2 | \boldsymbol \theta_1, \boldsymbol \theta_2, \boldsymbol \theta_3)}\frac{p(\boldsymbol \theta_2 | \boldsymbol \theta_1^\ast, \boldsymbol y_1)}{p(\boldsymbol \theta_2 | \boldsymbol \theta_1, \boldsymbol y_1)}, 1\right),\]
\item
  Update the parameter \(\boldsymbol \theta_3\) from the full conditional distribution \([\boldsymbol \theta_3 | \boldsymbol \theta_1, \boldsymbol \theta_2, \boldsymbol y_1, \boldsymbol y_2]\),
\end{enumerate}

recording the values of \(\boldsymbol \theta_1\), \(\boldsymbol \theta_2\), and \(\boldsymbol \theta_3\) at the end of each iteration.
\end{definition}

\begin{theorem}
\protect\hypertarget{thm:pprb-within-gibbs-sampling}{}\label{thm:pprb-within-gibbs-sampling}The PPRB-within-Gibbs sampler (Definition \ref{def:pprb-within-gibbs}) produces an ergodic Markov chain with the model's posterior distribution as its target distribution if the posterior distribution satisfies the following positivity condition,
\[p(\boldsymbol \theta_1 | \boldsymbol y_1, \boldsymbol y_2) > 0,\ p(\boldsymbol \theta_2 | \boldsymbol y_1, \boldsymbol y_2) > 0,\ p(\boldsymbol \theta_3 | \boldsymbol y_1, \boldsymbol y_2) > 0 \implies p(\boldsymbol \theta_1, \boldsymbol \theta_2, \boldsymbol \theta_3 | \boldsymbol y_1, \boldsymbol y_2) > 0.\]
\end{theorem}

\begin{proof}
See Appendix \ref{sec:sampler-definitions-theorems}.
\end{proof}

\(S\) is the number of samples drawn from the previous posterior distribution, \(p(\boldsymbol \theta_1 | \boldsymbol y_1)\) and generally cannot be increased. As the pool of samples, \(\{\boldsymbol \theta_1^s\}_{s=1}^S\), approximates the distribution \(p(\boldsymbol \theta_1 | \boldsymbol y_1)\) for the purpose of proposals, a larger \(S\) will lead to better proposals. However, we see in Section \ref{sec:pprb-degeneracy} that the pool of samples available to PPRB or PPRB-within-Gibbs will degrade over time after repeated applications in a streaming setting. A large \(S\) can extend the utility of the pool but will not keep it from degrading. We briefly mention future work that could address this degradation in Section \ref{sec:discussion}.

PPRB-within-Gibbs is applicable to the streaming record linkage model via the relationships \(\boldsymbol \theta_1 = \boldsymbol Z^{(1:(k-2))},\ \boldsymbol \theta_2 = [\boldsymbol m, \boldsymbol u],\ \boldsymbol \theta_3 = \boldsymbol Z^{(k-1)},\ \boldsymbol y_1 = \Gamma^{(1:(k-2))},\ \boldsymbol y_2 = \Gamma^{(k-1)},\) which satisfies all the preconditions of the algorithm. The algorithm steps for the streaming record linkage model as defined in Section \ref{sec:bayesian-record-linkage-model} are listed in Appendix \ref{sec:pprb-within-gibbs-streamingrl}. The acceptance ratio, \(\alpha\), is now the product of two ratios. The first ratio is of the data distribution of new data, as in original PPRB, evaluated both at the proposed \(\boldsymbol Z_\ast^{(1:(k-1))}\) and the current \(\boldsymbol Z^{(1:(k-1))}\). The second ratio is of the full conditional density of \(\boldsymbol m\) and \(\boldsymbol u\), but only conditioned on the pre-arrived data and pre-existing parameters. As such, these values can be pre-calculated for every existing posterior sample from the previous stage posterior.

This approach retains the appealing speed and low storage requirements of PPRB by utilizing existing posterior samples, while also avoiding an identified challenge of the original method proposed by \citet{hooten2019making} by drawing from the full conditional distribution of \(\boldsymbol Z^{(k-1)}\) rather than its prior. However, as resampling filtering methods, PPRB and PPRB-within-Gibbs can never sample states of any \(\boldsymbol Z^{(m)}, m < k\) not present in the first pool of posterior samples of that parameter. As a result, the pool of samples for any \(\boldsymbol Z^{(m)}\) will converge to a degenerate distribution as \(k \to \infty\) \citep{lunn2013fully}. We see evidence in Section \ref{sec:pprb-degeneracy} and discuss potential ways to address this in Section \ref{sec:discussion}.

\hypertarget{sequential-mcmc-smcmc}{%
\subsection{Sequential MCMC (SMCMC)}\label{sequential-mcmc-smcmc}}

Sequential MCMC is a sampling algorithm based on parallel sequential approximation \citep{yang2013sequential}. Starting from an existing ensemble of posterior samples from \(P(\boldsymbol m, \boldsymbol u, \boldsymbol Z^{(1:(k-2))} | \Gamma^{(1:(k-2))})\), SMCMC uses two kernels:

\begin{enumerate}
\def\labelenumi{\arabic{enumi}.}
\tightlist
\item
  The Jumping Kernel --- a probability distribution \(J(\boldsymbol Z^{(k-1)} | \cdot)\) which is responsible for initializing a value of \(\boldsymbol Z^{(k-1)}\) for each sample, potentially conditioning on old or new data.
\item
  The Transition Kernel --- any MCMC kernel, \(T\), that targets the updated posterior distribution, \(P(\boldsymbol m, \boldsymbol u, \boldsymbol Z^{(1:(k-1))} | \Gamma^{(1:(k-1))})\).
\end{enumerate}

These kernels are applied in parallel as initialized at each existing sample, first using the jumping kernel to initialize \(\boldsymbol Z^{(k-1)}\) and then repeatedly applying the transition kernel \(T\) until desired convergence is achieved. Final states of each parallel chain are taken as the new ensemble. SMCMC is a massively parallel MCMC algorithm that is expected to have fast convergence if the posterior based on new data and the posterior based on current data are similar in shape. Both jumping and transition kernels may depend on previously arrived data as well as new data. For Bayesian multi-file record linkage, we choose the transition kernel \(T\) as a Gibbs-style kernel which updates all parameters in sequence, and the jumping kernel \(J\) to be the full conditional update of \(\boldsymbol Z^{(k-1)}\).

SMCMC differs from PPRB-within-Gibbs in that it operates on an independent ensemble of samples. If the initial size of the ensemble is \(S\), SMCMC produces \(S\) \emph{independent} samples from the updated posterior distribution by nature of the parallel algorithm. Therefore the ensemble can remain relatively small, and only a small number of posterior draws need to be saved after the arrival of each file. The ensemble is never filtered, so converging to a degenerate distribution is not a concern for SMCMC. The transition kernel within SMCMC updates all parameters and maintains the same speed as MCMC for the updated posterior using the full data. The speed benefits of SMCMC then come from the ability to use as many as \(S\) parallel chains with well-chosen initial values. By contrast, PPRB-within-Gibbs's speed benefits come from simplifying the parameter update step. Unlike PPRB-within-Gibbs, SMCMC requires the full data be stored in perpetuity because with every new file the transition kernel will update all parameters.

\hypertarget{sec:z-proposals}{%
\subsection{Proposals for Matching Vector Updates}\label{sec:z-proposals}}

Both streaming samplers, PPRB-within-Gibbs and SMCMC, depend on full conditional updates of matching vectors. Step 3 of PPRB-within-Gibbs and the jumping kernel from SMCMC are both full conditional updates of the most recent vector \(\boldsymbol Z^{(k-1)}\), and the transition kernel of SMCMC must update all matching vectors. The choice of update is crucial for both speed and convergence of the sampler.

A straight-forward method for performing updates of \(\boldsymbol Z^{(k-1)}\) is to update each component \(Z^{(k-1)}_j\) in turn for \(j=1,\dots,n_k\). This method is used by \citet{sadinle2017bayesian} to update the matching vector in the two-file Bayesian record linkage model. The support for each component \(Z^{(k-1)}_j\) is enumerable as \(\{1, \dots, \sum_{\ell=1}^{k-1} n_\ell, \sum_{\ell=1}^{k-1} n_\ell + j\}\). To draw from the full conditional distribution of each \(Z^{(k-1)}_j\), the product of the likelihood and priors is evaluated for each potential value, normalized, and used as probabilities to sample the new value. The full transition kernel using these component-wise proposals for matching vectors is defined in Definition \ref{def:gibbs-sampler-component} in Appendix \ref{sec:sampler-definitions-theorems}.

\citet{zanella2020informed} describes a class of locally balanced pointwise informed proposals distributions to improve sampling in high-dimensional discrete spaces. For a sample space \({\cal X}\) with a target distribution given by \(\pi(x)\), these proposals have the form,
\[Q_g(x, y) = \frac{g\left(\frac{\pi(y)}{\pi(x)}\right)K(x,y)}{Z_g(x)},\]
where the proposed move is from a point \(x\) to a point \(y\). \(K(x,y)\) is a symmetric uninformed local proposal distribution, \(g: \mathbb{R}^+ \to \mathbb{R}^+\) is a function and \(Z_g(x)\) is the normalizing constant. The goal of these proposals is to improve the uninformed proposal \(K\) by biasing towards points with higher probability through the multiplicative term \(g(\pi(y)/\pi(x))\). The uninformed kernel \(K\) is arbitrary, and \(Q_g\) is called \emph{locally balanced} if and only if \(g(t) = tg(1/t)\). A consequence of this property of \(g\) along with a symmetric local proposal \(K\) is that the Metropolis-Hastings acceptance ratio for locally balanced proposals simplifies to the ratio of normalizing constants, \(\min(Z_g(x)/Z_g(y), 1)\).

To apply locally balanced proposals to the Bayesian multi-file record linkage model, we choose \(g(t) = t/(1+t)\) and \(K\) to be the kernel defined by making a single randomly chosen add, delete, swap, or double-swap move. The kernel \(K\) can optionally be blocked, where first a subset of records in file \(X_k\) and an equally sized subset of records in files \(X_1, \dots, X_{k-1}\) are randomly selected and then, only moves which affect links between these subsets are considered. Blocking limits the scope of possible moves for each update, which in turn decreases the time required per update. However, blocking also increases the chance of proposing a move to a lower probability state which is more likely to be rejected, requiring more updates to sample effectively. We use a block size in Section \ref{sec:simulation-study} which is fast while still producing many accepted proposals. The full transition kernel using these locally balanced proposals for matching vectors is defined in Definition \ref{def:gibbs-sampler-lb} in Appendix \ref{sec:sampler-definitions-theorems}.

The component-wise full conditional updates can take larger steps than the locally balanced proposals because each value in \(\boldsymbol Z^{(k-1)}\) has the potential to be updated. In contrast, the locally balanced proposals can at most update two components of \(\boldsymbol Z^{(k-1)}\) with a double-swap operation. The component-wise full conditional updates, however, are more computationally intensive as the likelihood needs to be calculated at more potential states and there is no option for blocking. We use locally balanced proposals to update \(\boldsymbol Z^{(k-1)}\) in PPRB-within-Gibbs. In Section \ref{sec:simulation-study} both locally balanced and component-wise proposals are used within SMCMC and their speed and sampling performance are compared.

\hypertarget{sec:simulation-study}{%
\section{Simulation Study}\label{sec:simulation-study}}

To assess both the performance of the model and speed of the streaming update, we evaluate our Bayesian multi-file record linkage model and both streaming samplers on simulated data. We choose to focus on the four file case, since the arrival of the fourth file is the earliest point at which two sequential streaming updates can have been used, demonstrating the potential for use in streaming settings.

\hypertarget{sec:data-simulation}{%
\subsection{Data Simulation}\label{sec:data-simulation}}

Data were simulated using the GeCo software package \citep{tran2013geco} which creates realistic simulated data about individuals. Each record was given 10 fields: first name, last name, occupation, and age, plus 6 categorical fields with values drawn uniformly from 12 possible categories. For each of four levels of overlap (10\%, 30\%, 50\%, and 90\%), four files of 200 records each were created. Duplicate records were allowed in consecutive and non-consecutive datasets. In each duplicated record in files \(X_2\), \(X_3\) and \(X_4\), a maximum of either 2, 4, or 6 errors were inserted. Errors were inserted into text fields of first name and last name by simulating typos, common misspellings, and OCR errors using the GeCo package \citep{tran2013geco}. Errors were inserted into the remaining categorical fields by replacing their value with a category selected randomly uniform from all possible categories. Each field could have errors, with text fields more likely than categorical fields. A total of 12 datasets were created, one at each combination of error and overlap. This simulation is intended to mimic a longitudinal survey in which we have demographic information and the answers to 6 identifying categorical questions with varying levels of noise and overlap. Comparison vectors were created by comparing each field between pairs of records. Text fields were compared using binned normalized Levenshtein distance with 4 levels: exact equality, \((0,0.25]\), \((0.25,0.5]\), and \((0.5, 1]\). Categorical fields were compared in a binary fashion. All computation in this section and in Section \ref{sec:poland} was performed using the RMACC Summit Supercomputer \citep{anderson2017deploying}. We utilized the accompanying package \texttt{bstrl} \citep{pkgbstrl} on \texttt{R} version 3.5.0 on Intel Haswell CPUs with 24 cores and 4.84 GB of memory per CPU.

\hypertarget{sec:simulation-link-accuracy}{%
\subsection{Link Accuracy}\label{sec:simulation-link-accuracy}}

We assess the accuracy of our multi-file record linkage model by evaluating samples from the posterior distribution obtained using a non-streaming Gibbs sampler. The streaming samplers should target the same posterior distribution as the Gibbs sampler, thus we present a comparison on model performance alone. We compare the streaming samplers on runtime in Section \ref{sec:simulation-sampler-speed}. We use three strengths of prior on the parameter \(\boldsymbol m\). For the diffuse prior (Flat), we set \(\boldsymbol a=\begin{bmatrix}1&\cdots &1\end{bmatrix}\). Then for weakly informed (Weak) and strongly informed (Strong) priors, we use Equation \ref{eqn:m-prior-construction} to determine \(\boldsymbol a\). We use \(s=12\) for the weakly informed prior and \(s=120\) for the strongly informed prior. In both the weakly and strongly informed priors, \(p=1/2\) for string fields and \(p=1/8\) for categorical fields. These values of \(p\) reflect a prior probability of error of \(1/2\) in string fields and \(1/8\) in categorical fields, and an average of 2 errors per record. For comparison, we evaluate the multi-file Bayesian linkage model of \citet{aleshinguendel2021multifile} as implemented in the \texttt{multilink} package (Multilink), the empirically motivated Bayesian entity resolution model of \citet{steorts2015entity} as implemented in the \texttt{blink} package (Blink), and a semi-supervised Fellegi-Sunter model with support vector machine used to classify links as implemented in the \texttt{RecordLinkage} package (SVM) with 1\% of the record pairs used as training data. Multilink is similar to our proposed model in that it is a Bayesian multi-file Fellegi-Sunter extension. However, it differs from ours in that it is based on a partitioning prior and does not enable streaming data. We have included both the recommended separate likelihoods, which models comparisons differently for each pair of files, and a single likelihood version (Single Likelihood), which is more analogous to the model presented in Section \ref{sec:streaming-likelihood}. Blink and SVM are both deduplication models, and so may link records within the same file. Where possible, we have chosen default or recommended values for tuning and hyperparameters in these comparison models. Further details about the comparison models can be found in Appendix \ref{sec:link-accuracy-model-details}.

We compare the accuracy of the resulting links by examining the posterior distribution of the \(F_1\)-score, \(F_1 = 2(\text{recall}^{-1} + \text{precision}^{-1})^{-1}\) \citep{blair1979information}. Recall is the proportion of true coreferent record pairs that are correctly identified, and precision is the proportion of identified coreferent pairs that are true duplicates. Table \ref{tab:f1score} shows these posterior distributions as means and standard deviations of posterior samples drawn from each model, after discarding burn-in. We also evaluate the models through the posterior distribution of the number of estimated distinct entities across all files in Figure \ref{fig:entityerrors}. Because the SVM may result in non-transitive links, we consider only the accuracy of the link labels for this method rather than number of estimated entities. The model presented in this paper performs as well or better than the comparison models using both metrics. Additional error levels are included in the supplemental material.

Overall the link accuracy of our model is comparable to the comparison models. In all but one case (90\% overlap and 6 errors) our proposed model has the highest \(F_1\)-score, and in that case our model's \(F_1\)-score is close to the best-performing comparison model. As expected, performance is generally worse for all models in scenarios with fewer duplicates and more errors in the duplicates. We would hesitate to generalize these comparison results to other scenarios, particularly because two comparison models (Blink, SVM) allow for duplicates within files which are not present in this simulated data. Additionally, the SVM method relies on having training data, which is not always available and expensive to produce, while the proposed model is fully unsupervised. With higher amounts of error and low overlap, the strength of the prior on \(\boldsymbol m\) can be used to compensate for a lack of clean identifying information. We see in these cases, that the Strong Prior model outperforms the Weak and Flat Prior models, even though the strong prior is slightly misspecified for higher error cases. Similar prior information may be provided for the other Bayesian comparison models (Blink, Multilink), which may also improve their performance in these more difficult cases. Each Bayesian model was run using 3 different random seeds and all exhibited some multimodality in higher overlap cases where links are more constrained, particularly those with duplicate-free file constraints (Streaming, Multilink).

\newcommand{\originalarraystretch}{}
\let\originalarraystretch\arraystretch
\renewcommand{\arraystretch}{0.80}
\begin{table}

\caption{\label{tab:f1score}Posterior means and standard deviations of $F_1$-score for simulated datasets. Within rows, each model is listed: the model presented in this paper (Streaming) and three comparison models. Larger values represent more accurate links in the posterior distribution. The support vector machine, a non-bayesian method, is represented only by the $F_1$-score of its resulting point estimate.}
\centering
\fontsize{10}{12}\selectfont
\begin{tabular}[t]{lllll}
\toprule
\begingroup\fontsize{10}{12}\selectfont Model\endgroup & \begingroup\fontsize{10}{12}\selectfont 10\% overlap\endgroup & \begingroup\fontsize{10}{12}\selectfont 30\% overlap\endgroup & \begingroup\fontsize{10}{12}\selectfont 50\% overlap\endgroup & \begingroup\fontsize{10}{12}\selectfont 90\% overlap\endgroup\\
\midrule
\addlinespace[0.3em]
\multicolumn{5}{l}{\textbf{Errors: 2}}\\
\hspace{1em}Streaming (Flat Prior) & \textbf{0.992 (0.0054)} & \textbf{1.000 (0.0009)} & 0.991 (0.0018) & 0.990 (0.0000)\\
\hspace{1em}Streaming (Weak Prior) & \textbf{0.992 (0.0056)} & \textbf{1.000 (0.0009)} & \textbf{0.999 (0.0015)} & \textbf{1.000 (0.0000)}\\
\hspace{1em}Streaming (Strong Prior) & 0.978 (0.0102) & 0.999 (0.0022) & 0.994 (0.0020) & \textbf{1.000 (0.0000)}\\
\hspace{1em}Multilink & 0.985 (0.0089) & 0.996 (0.0041) & 0.985 (0.0019) & 0.944 (0.0000)\\
\hspace{1em}Multilink (Single Likelihood) & 0.991 (0.0047) & 0.999 (0.0016) & 0.994 (0.0015) & 0.992 (0.0000)\\
\hspace{1em}Blink & 0.578 (0.0165) & 0.974 (0.0021) & 0.993 (0.0005) & 0.996 (0.0004)\\
\hspace{1em}SVM (1\% training) & 0.962 & \textbf{1.000} & 0.986 & 0.999\\
\addlinespace[0.3em]
\multicolumn{5}{l}{\textbf{Errors: 4}}\\
\hspace{1em}Streaming (Flat Prior) & 0.979 (0.0123) & 0.957 (0.0067) & 0.974 (0.0036) & 0.997 (0.0001)\\
\hspace{1em}Streaming (Weak Prior) & \textbf{0.981 (0.0107)} & 0.971 (0.0072) & \textbf{0.986 (0.0034)} & \textbf{0.998 (0.0001)}\\
\hspace{1em}Streaming (Strong Prior) & 0.978 (0.0101) & \textbf{0.976 (0.0052)} & \textbf{0.986 (0.0036)} & \textbf{0.998 (0.0001)}\\
\hspace{1em}Multilink & 0.161 (0.0038) & 0.640 (0.0402) & 0.982 (0.0048) & 0.978 (0.0015)\\
\hspace{1em}Multilink (Single Likelihood) & 0.913 (0.0283) & 0.960 (0.0092) & 0.983 (0.0035) & 0.997 (0.0004)\\
\hspace{1em}Blink & 0.504 (0.0117) & 0.887 (0.0065) & 0.962 (0.0043) & 0.994 (0.0011)\\
\hspace{1em}SVM (1\% training) & 0.933 & 0.827 & 0.919 & 0.947\\
\addlinespace[0.3em]
\multicolumn{5}{l}{\textbf{Errors: 6}}\\
\hspace{1em}Streaming (Flat Prior) & 0.227 (0.0073) & 0.797 (0.0200) & 0.952 (0.0071) & 0.993 (0.0016)\\
\hspace{1em}Streaming (Weak Prior) & 0.808 (0.0592) & 0.910 (0.0157) & \textbf{0.954 (0.0065)} & 0.977 (0.0011)\\
\hspace{1em}Streaming (Strong Prior) & \textbf{0.896 (0.0180)} & \textbf{0.929 (0.0103)} & 0.952 (0.0054) & 0.983 (0.0012)\\
\hspace{1em}Multilink & 0.064 (0.0013) & 0.482 (0.0118) & 0.822 (0.0263) & 0.985 (0.0017)\\
\hspace{1em}Multilink (Single Likelihood) & 0.064 (0.0021) & 0.393 (0.0151) & 0.913 (0.0147) & \textbf{0.997 (0.0012)}\\
\hspace{1em}Blink & 0.456 (0.0127) & 0.803 (0.0092) & 0.910 (0.0058) & 0.986 (0.0022)\\
\hspace{1em}SVM (1\% training) & 0.674 & 0.668 & 0.707 & 0.675\\
\bottomrule
\end{tabular}
\end{table}
\renewcommand{\arraystretch}{\originalarraystretch}

\begin{figure}
\centering
\includegraphics{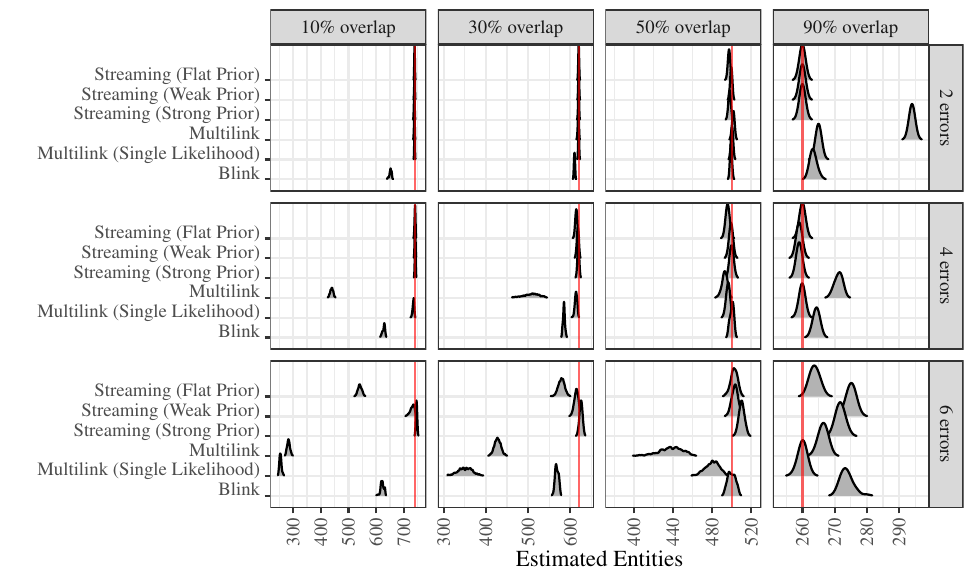}
\caption{\label{fig:entityerrors}Posterior distribution of the number of estimated entities for simulated datasets. A vertical line indicates the true number of distinct entities in each dataset. Compared models are on the y-axis: the model presented in this paper (Streaming) and three comparison models.}
\end{figure}

\hypertarget{sec:simulation-sampler-speed}{%
\subsection{Speed}\label{sec:simulation-sampler-speed}}

\begin{figure}
\centering
\includegraphics{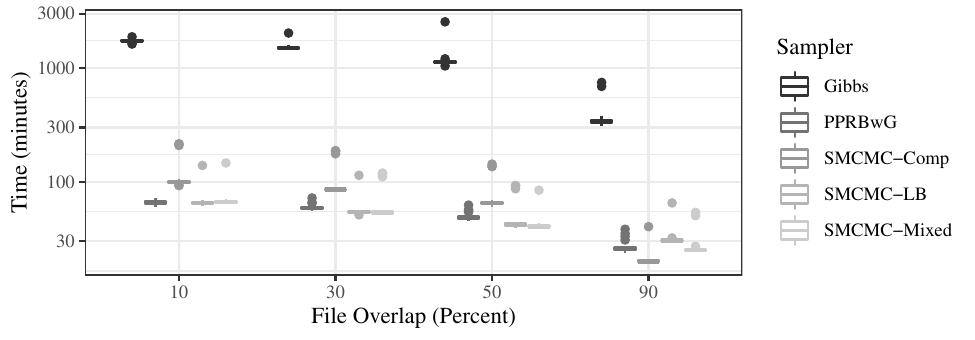}
\caption{\label{fig:esstime}Time required for each sampler to produce an effective sample size of 1000. The effective sample size is measured on the continuous parameters, \(\boldsymbol m\) and \(\boldsymbol u\). Lower values indicate more efficient sampling. The SMCMC sampling time is estimated assuming 1000 available cores so that each ensemble member can be updated in parallel.}
\end{figure}

Our streaming samplers from Section \ref{sec:streaming-sampling} more efficiently produce samples from the model's posterior distribution. We demonstrate this improved efficiency by recording the amount of time required by each sampler to produce an effective sample size of 1000. For each of the 16 simulated data sets, five samplers were used to sample from the posterior distribution of \(\boldsymbol m, \boldsymbol u, \boldsymbol Z^{(1)}, \boldsymbol Z^{(2)}, \boldsymbol Z^{(3)} | \Gamma^{(1)}, \Gamma^{(2)}, \Gamma^{(3)}\). We compared PPRB-within-Gibbs using locally balanced \(\boldsymbol Z^{(3)}\) updates (PPRBwG), SMCMC with locally balanced proposals for both jumping and transition kernels (SMCMC-LB), SMCMC with component-wise full conditional draws for both jumping and transition kernels (SMCMC-Comp), SMCMC with component-wise full conditional draws for the jumping kernel and locally balanced proposals for the transition kernel (SMCMC-Mixed), and a non-streaming Gibbs sampler fit to the full data using the sampler in Definition \ref{def:gibbs-sampler-component} in Appendix \ref{sec:sampler-definitions-theorems} (Gibbs). All streaming samplers used the \texttt{BRL} package \citep{sadinle2017bayesian} to sample from the bipartite record linkage posterior distribution, \(\boldsymbol m, \boldsymbol u, \boldsymbol Z^{(1)} | \Gamma^{(1)}\). More details about these simulations are in Appendix \ref{sec:speed-comparison-sampler-details}.

We choose effective sample size to capture both the number of samples produced in a given time and their quality. To summarize the effective sample size of each run, we calculate the effective sample size of each component of the continuous parameters \(\boldsymbol m\) and \(\boldsymbol u\), and find the median across all values. Since SMCMC produces independent samples, the effective sample size of any parameter is equal to the size of the SMCMC ensemble. The three SMCMC methods are assumed to be run fully parallel, where the samples produced are not limited by time but by available computational resources. The streaming samplers take an order of magnitude less time to obtain 1000 effective samples than the non-streaming sampler (Figure \ref{fig:esstime}). With fewer cores available the time advantage for SMCMC will not be as stark, however there is still a benefit with as few as 36 cores.

As the number of records in each file, \(n_1, \dots, n_k\), grows, the time required by each component-wise \(\boldsymbol Z^{(k-1)}\) full conditional update will grow quadratically because it iterates through every combination of a record in file \(X_k\) and a record in all previous files, \(n_k \cdot \sum_{\ell=1}^{k-1} n_\ell\) total pairs of records. This will affect the time of any sampler using component-wise full conditional updates. The time for locally balanced proposals, if blocked, does not grow with the number of records per file. However the smaller the block size becomes relative to the file size, the less effective blocked locally balanced proposals will be at exploring the parameter space. As the number of files, \(k\), grows, the time required by each component-wise \(\boldsymbol Z^{(k-1)}\) full conditional update will grow linearly since the number of records in file \(k\) does not increase, only the total number of records in previous files. The time for locally balanced proposals, if blocked, does not grow with the number of files. A growing number of files will also increase the time required by SMCMC as more full conditional updates will be required per iteration of the transition kernel. The time required for the transition kernel will grow at most quadratically with increasing \(k\) because a linear series of new full conditional updates are required which are themselves require at most linearly increasing time with \(k\). As \(k\) increases, the amount of time required for PPRB-within-Gibbs is not affected unless using component-wise full conditional updates for \(\boldsymbol Z^{(k-1)}\) or also increasing the locally balanced proposal block-size.

\hypertarget{sec:pprb-degeneracy}{%
\subsection{PPRB Degeneracy}\label{sec:pprb-degeneracy}}

\begin{figure}
\centering
\includegraphics{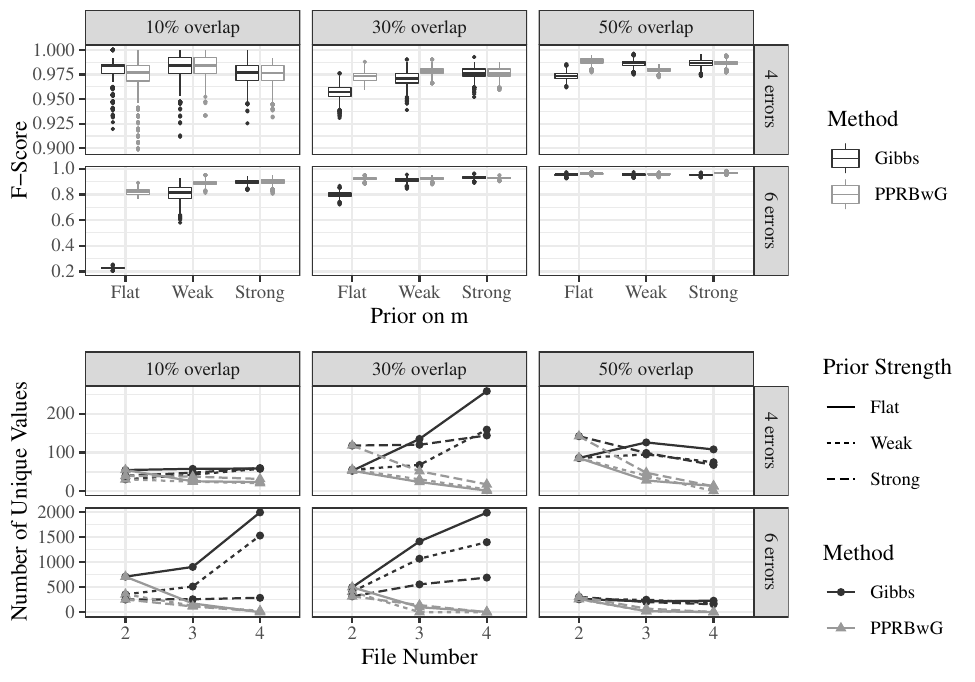}
\caption{\label{fig:pprbdegen}Demonstrations of PPRB-within-Gibbs sample degredation. The data scenarios in which the difference in distinct values is most visible are shown. TOP: Posterior F1-score for PPRB-within-Gibbs and non-streaming samplers. On the x-axis are different strengths of prior distribution on the parameter \(\boldsymbol m\). In some datasets, PPRB-within-Gibbs appears to produce different posterior distributions than the non-streaming sampler. BOTTOM: Number of distinct values of \(\boldsymbol Z^{(1)}\) produced by PPRB-within-Gibbs and Gibbs, out of 2000 iterations. Lines connect points with the same prior information for \(\boldsymbol m\).}
\end{figure}

As is true of all filtering methods, PPRB and PPRB-within-Gibbs have the undesirable property that the pool of samples for any \(\boldsymbol Z^{(m)}\) will converge to a degenerate distribution as \(k \to \infty\). We see an example of this phenomenon in Figure \ref{fig:pprbdegen}, particularly for overlaps of 50\% or less and 4 or more errors, where the posterior distribution of \(F_1\)-score from PPRB-within-Gibbs differs from the other samplers. For 10\% overlap, 4 errors, and a flat prior on \(\boldsymbol m\), we even see a very large difference between PPRB-within-Gibbs and the non-streaming sampler. To investigate further, we compare the samples produced from PPRB-within-Gibbs to those produced from a non-streaming (Gibbs) sampler. In 4-file record linkage, PPRB-within-Gibbs produces noticeably fewer unique values of \(\boldsymbol Z^{(1)}\) than Gibbs for the same number of posterior samples. This indicates that degradation is occurring due to the filtering of the initial pool of samples from two sequential PPRB-within-Gibbs updates. As more files are added and the pool of \(\boldsymbol Z^{(1)}\) samples is further filtered, this contrast will become more apparent, eventually leading to a single value of \(\boldsymbol Z^{(1)}\) being sampled.

\hypertarget{sec:poland}{%
\section{Real Data Application}\label{sec:poland}}

We now apply streaming record linkage to a sample of records from a longitudinal survey with a known true identity for each record. The Social Diagnosis Survey (SDS) of quality of life in Poland \citep{diagnoza} is a biennial survey of households that was first conducted in the year 2000. Individuals may be recorded multiple times in separate years but there is no duplication of individuals within a year. Four files of data were selected from the full dataset from the years 2007 through 2013. The four files have varying sizes, with \(n_1=151\), \(n_2=464\), \(n_3=688\), and \(n_4=677\), for a total of 1980 records. The files were created by randomly sampling, without replacement, 910 individuals from all individuals appearing in at least one of the included years. Of the 910 individuals, 306 appear in just one file, 240 appear in two files, 262 appear in three files, and 102 appear in all four files.

Linkage was performed using six fields: gender, province, educational attainment, and year, month, and day of birth. All fields are categorical and were compared using binary comparisons. We chose hyperparameters to produce flat priors in \(\boldsymbol m\), \(\boldsymbol u\), and \(\boldsymbol Z^{(\ell)}\) for \(\ell=1,2,3\). We compared five samplers: a non-streaming Gibbs sampler (Gibbs), sequentially applied PPRB-within-Gibbs updates with locally balanced proposals (PPRBwG), and sequentially applied SMCMC updates with component-wise proposals (SMCMC-Comp), locally balanced proposals (SMCMC-LB) or a mix using component-wise jumping kernel proposals and locally balanced transition kernel proposals (SMCMC-Mixed). More details of the MCMC runs can be found in Appendix \ref{sec:poland-mcmc-details}.

\begin{table}

\caption{\label{tab:poland-accuracy}Posterior means and standard deviations of $F_1$-score and estimated number of entities, and total sampling time, for the four-file Poland SDS data set using five samplers. There are 910 true entities in the four files. Sampling time is given in cumulative hours required to produce posterior samples of the parameters conditioned first on three files, then on four files using each sampling method. The SMCMC sampling time is estimated assuming 1000 available cores so that each ensemble member can be updated in parallel.}
\centering
\begin{tabular}[t]{lllr}
\toprule
Sampler & F1-Score & Estimated Entities & Sampling Time\\
\midrule
Gibbs & 0.985  (9e-04) & 915  (1.6) & 121.1\\
PPRBwG & 0.992  (0.0010) & 915  (2.0) & 10.9\\
SMCMC-Comp & 0.992  (0.0012) & 916  (1.9) & 3.5\\
SMCMC-LB & 0.99  (0.0022) & 916  (1.9) & 6.9\\
SMCMC-Mixed & 0.992  (0.0010) & 916  (1.8) & 3.5\\
\bottomrule
\end{tabular}
\end{table}

The streaming record linkage models were able to recover the true coreferent records with high accuracy. Table \ref{tab:poland-accuracy} shows the posterior \(F_1\)-score distribution for each of the 5 samplers, the posterior distribution of the estimated number of entities resulting from the linkage, and the time to generate the posterior samples. All samplers performed equally well at recovering the true coreferent record sets with a posterior mean \(F_1\)-score between 0.985 and 0.992. Streaming samplers were significantly faster than the non-streaming Gibbs sampler, with times given for the cumulative time required to produce both three-file and four-file inference using each sampling method. This is representative of the streaming data setting where inference is required after each new file arrives. The streaming samplers show between 11 times and 35 times speedup when compared to the non-streaming Gibbs sampler, where SMCMC time estimates are based on the assumption that enough cores are available for each ensemble to be run simultaneously in parallel.

\hypertarget{sec:discussion}{%
\section{Discussion}\label{sec:discussion}}

In this paper we have introduced a model for multi-file Bayesian record linkage based on the Fellegi-Sunter paradigm that is appropriate for streaming data contexts. We have shown this model to work as well as comparison models on realistic simulated data at varying amounts of duplication and error. With this work, we have proposed the first model-based streaming record linkage procedures that update inference on existing parameters and estimate new parameters as new data arrives. Our model provides interpretable parameters for estimating not only links between records, but the probability of different levels of error between fields of coreferent records. These streaming samplers allow for near-identical inference to the model fit using the full data. Having two distinct streaming options for this model allows for the selection of one based on the needs of the user, and we have detailed the trade-offs that one might consider. We have demonstrated that these streaming samplers can provide significant computational gains when compared to a Gibbs sampler using both simulated and real-world data.

Our simulation study shows a noticeable effect of the strength of the prior on \(\boldsymbol m\) on the accuracy of the resulting posterior samples. In Section \ref{sec:m-u-prior} we describe a way to use the prior on \(\boldsymbol m\) to incorporate prior knowledge about the probability of errors in duplicated fields, and in Section \ref{sec:simulation-link-accuracy} we suggest how this can be used to compensate for a lack of clean identifying fields in each record. The priors on \(\boldsymbol Z^{(1:(k-1))}\) can also be tuned through the values of \(\alpha_\pi\) and \(\beta_\pi\), but practitioners are unlikely to have prior knowledge about the level of overlap between files.

The scalability of this model to files with very large numbers of records could be limited in two ways. First, the dimension of the model's parameter space grows directly with the number of records included. A larger parameter space requires both larger storage for posterior samples and slower computation of the transition kernel. A very large file also poses difficulties for the computation of comparisons. With the arrival of a new file, a comparison vector needs to be computed comparing each record in the new file to each record in previous files. These challenges with large files could be mitigated by blocking to prohibit links across large time differences or breaking large files into several smaller files.

Future work in streaming record linkage includes relaxing the assumption of no duplicates within files to develop an entity resolution model that can identify duplicates both within and between files in a streaming context. Further streaming sampling methods may be explored by combining techniques of PPRB-within-Gibbs and SMCMC into a streaming sampler with more of the strengths of both methods: the ease of computation and low data storage demands of PPRB with the non-degenerate sampling of SMCMC. Additionally, more informative prior distribution selection for the linkage parameters is an area of future research.

\hypertarget{acknowledgements}{%
\section*{Acknowledgements}\label{acknowledgements}}
\addcontentsline{toc}{section}{Acknowledgements}

The authors acknowledge the support of the Laboratory for Analytic Sciences at North Carolina State University. B. Betancourt acknowledges the support of NSF DMS-2310222.

\hypertarget{supplemental-materials}{%
\section*{Supplemental Materials}\label{supplemental-materials}}
\addcontentsline{toc}{section}{Supplemental Materials}

All supplemental materials are contained in a single compressed (zipped) archive.

\begin{description}

\item[Appendix to ``Fast Bayesian Record Linkage for Streaming Data Contexts'':] Appendices that include supplemental tables and figures; posterior and full conditional distributions; supplemental definitions, theorems, and proofs; and simulation details. (streaming-record-linkage-appendix.pdf, PDF document)

\item[R-package for streaming record linkage:] R-package `bstrl`, implementing the streaming record linkage model and the PPRB-within-Gibbs and SMCMC streaming updates. (bstrl\_1.0.2.tar.gz, GNU zipped tar file)

\item[Reproducible code repository:] R code that can be used to reproduce the numerical results in this article, including tables and figures (streamingrl-reproducible-main.zip, compressed folder)

\end{description}

\bibliographystyle{Chicago}
\bibliography{refs.bib}

\begin{thebibliography}{}

\bibitem[\protect\citeauthoryear{Aleshin-Guendel and Sadinle}{Aleshin-Guendel
  and Sadinle}{2022}]{aleshinguendel2021multifile}
Aleshin-Guendel, S. and M.~Sadinle (2022).
\newblock {Multifile Partitioning for Record Linkage and Duplicate Detection}.
\newblock {\em Journal of the American Statistical Association\/}~{\em 0\/}(0),
  1--10.

\bibitem[\protect\citeauthoryear{Altwaijry, Kalashnikov, and
  Mehrotra}{Altwaijry et~al.}{2017}]{Altwaijry17}
Altwaijry, H., D.~V. Kalashnikov, and S.~Mehrotra (2017).
\newblock {QDA: A Query-Driven Approach to Entity Resolution}.
\newblock {\em IEEE Transactions on Knowledge and Data Engineering\/}~{\em
  29\/}(2), 402--417.

\bibitem[\protect\citeauthoryear{Anderson, Burns, Milroy, Ruprecht, Hauser, and
  Siegel}{Anderson et~al.}{2017}]{anderson2017deploying}
Anderson, J., P.~J. Burns, D.~Milroy, P.~Ruprecht, T.~Hauser, and H.~J. Siegel
  (2017).
\newblock {Deploying RMACC Summit: An HPC Resource for the Rocky Mountain
  Region}.
\newblock In {\em Proceedings of the Practice and Experience in Advanced
  Research Computing 2017 on Sustainability, Success and Impact}, PEARC17, New
  York, NY, USA. Association for Computing Machinery.

\bibitem[\protect\citeauthoryear{Betancourt, Zanella, Miller, Wallach, Zaidi,
  and Steorts}{Betancourt et~al.}{2016}]{betancourt2016flexible}
Betancourt, B., G.~Zanella, J.~W. Miller, H.~Wallach, A.~Zaidi, and R.~C.
  Steorts (2016).
\newblock {Flexible Models for Microclustering with Application to Entity
  Resolution}.
\newblock In D.~Lee, M.~Sugiyama, U.~Luxburg, I.~Guyon, and R.~Garnett (Eds.),
  {\em {Advances in Neural Information Processing Systems}}, Volume~29. Curran
  Associates, Inc.

\bibitem[\protect\citeauthoryear{Binette and Steorts}{Binette and
  Steorts}{2022}]{binette2021some}
Binette, O. and R.~C. Steorts (2022).
\newblock {(Almost) all of entity resolution}.
\newblock {\em Science Advances\/}~{\em 8\/}(12), eabi8021.

\bibitem[\protect\citeauthoryear{Blair}{Blair}{1979}]{blair1979information}
Blair, D.~C. (1979).
\newblock {Information Retrieval, 2nd ed. C.J. Van Rijsbergen. London:
  Butterworths; 1979: 208 pp. Price: \$32.50}.
\newblock {\em Journal of the American Society for Information Science\/}~{\em
  30\/}(6), 374--375.

\bibitem[\protect\citeauthoryear{Christen}{Christen}{2012}]{christen2012data}
Christen, P. (2012).
\newblock {\em {Data matching concepts and techniques for record linkage,
  entity resolution, and duplicate detection }}.
\newblock {Data-centric systems and applications.} Berlin ;: Springer.

\bibitem[\protect\citeauthoryear{Christen, Gayler, and Hawking}{Christen
  et~al.}{2009}]{Christen09}
Christen, P., R.~Gayler, and D.~Hawking (2009).
\newblock {Similarity-Aware Indexing for Real-Time Entity Resolution}.
\newblock In {\em {Proceedings of the 18th ACM Conference on Information and
  Knowledge Management}}, CIKM '09, New York, NY, USA, pp.\  1565–1568.
  Association for Computing Machinery.

\bibitem[\protect\citeauthoryear{Czapinski and Panek}{Czapinski and
  Panek}{2015}]{diagnoza}
Czapinski, J. and T.~Panek (2015).
\newblock {Social diagnosis - objective and subjective quality of life in
  Poland}.
\newblock \url{http://www.diagnoza.com/index-en.html}.
\newblock Accessed: 2022-06-09.

\bibitem[\protect\citeauthoryear{Dey, Mookerjee, and Liu}{Dey
  et~al.}{2011}]{Dey11online}
Dey, D., V.~Mookerjee, and D.~Liu (2011).
\newblock {Efficient Techniques for Online Record Linkage}.
\newblock {\em IEEE Transactions on Knowledge and Data Engineering\/}~{\em
  23\/}(3), 373--387.

\bibitem[\protect\citeauthoryear{Fellegi and Sunter}{Fellegi and
  Sunter}{1969}]{fellegi1969theory}
Fellegi, I.~P. and A.~B. Sunter (1969).
\newblock {A Theory for Record Linkage}.
\newblock {\em Journal of the American Statistical Association\/}~{\em
  64\/}(328), 1183--1210.

\bibitem[\protect\citeauthoryear{Fleming, Kirby, and Penny}{Fleming
  et~al.}{2012}]{fleming2012record}
Fleming, M., B.~Kirby, and K.~I. Penny (2012).
\newblock {Record linkage in Scotland and its applications to health research}.
\newblock {\em Journal of Clinical Nursing\/}~{\em 21\/}(19pt20), 2711--2721.

\bibitem[\protect\citeauthoryear{Gutman, Afendulis, and Zaslavsky}{Gutman
  et~al.}{2013}]{gutman_2013}
Gutman, R., C.~C. Afendulis, and A.~M. Zaslavsky (2013).
\newblock {A Bayesian Procedure for File Linking to Analyze End-of-Life Medical
  Costs}.
\newblock {\em Journal of the American Statistical Association\/}~{\em
  108\/}(501), 34--47.
\newblock PMID: 23645944.

\bibitem[\protect\citeauthoryear{Hof, Ravelli, and Zwinderman}{Hof
  et~al.}{2017}]{HofRavelliZwinderman17}
Hof, M.~H., A.~C. Ravelli, and A.~H. Zwinderman (2017).
\newblock {A Probabilistic Record Linkage Model for Survival Data}.
\newblock {\em Journal of the American Statistical Association\/}~{\em
  112\/}(520), 1504--1515.

\bibitem[\protect\citeauthoryear{Hooten, Johnson, and Brost}{Hooten
  et~al.}{2021}]{hooten2019making}
Hooten, M.~B., D.~S. Johnson, and B.~M. Brost (2021).
\newblock {Making Recursive Bayesian Inference Accessible}.
\newblock {\em The American Statistician\/}~{\em 75\/}(2), 185--194.

\bibitem[\protect\citeauthoryear{Ioannou, Nejdl, Nieder\'{e}e, and
  Velegrakis}{Ioannou et~al.}{2010}]{Ioannou2010OntheFlyEQ}
Ioannou, E., W.~Nejdl, C.~Nieder\'{e}e, and Y.~Velegrakis (2010, sep).
\newblock {On-the-Fly Entity-Aware Query Processing in the Presence of
  Linkage}.
\newblock {\em Proc. VLDB Endow.\/}~{\em 3\/}(1–2), 429–438.

\bibitem[\protect\citeauthoryear{Kaplan, Betancourt, and Steorts}{Kaplan
  et~al.}{2022}]{kaplan2022practical}
Kaplan, A., B.~Betancourt, and R.~C. Steorts (2022).
\newblock {A Practical Approach to Proper Inference with Linked Data}.
\newblock {\em The American Statistician\/}~{\em 0\/}(0), 1--10.

\bibitem[\protect\citeauthoryear{Karapiperis, Gkoulalas{-}Divanis, and
  Verykios}{Karapiperis et~al.}{2018}]{Karapiperis2018}
Karapiperis, D., A.~Gkoulalas{-}Divanis, and V.~S. Verykios (2018).
\newblock {Summarization Algorithms for Record Linkage}.
\newblock In M.~H. B{\"{o}}hlen, R.~Pichler, N.~May, E.~Rahm, S.~Wu, and
  K.~Hose (Eds.), {\em {Proceedings of the 21st International Conference on
  Extending Database Technology, {EDBT} 2018, Vienna, Austria, March 26-29,
  2018}}, pp.\  73--84. OpenProceedings.org.

\bibitem[\protect\citeauthoryear{Lunn, Barrett, Sweeting, and Thompson}{Lunn
  et~al.}{2013}]{lunn2013fully}
Lunn, D., J.~Barrett, M.~Sweeting, and S.~Thompson (2013).
\newblock {Fully Bayesian hierarchical modelling in two stages, with
  application to meta-analysis}.
\newblock {\em Journal of the Royal Statistical Society: Series C (Applied
  Statistics)\/}~{\em 62\/}(4), 551--572.

\bibitem[\protect\citeauthoryear{Marchant, Kaplan, Elazar, Rubinstein, and
  Steorts}{Marchant et~al.}{2021}]{marchant2021dblink}
Marchant, N.~G., A.~Kaplan, D.~N. Elazar, B.~I.~P. Rubinstein, and R.~C.
  Steorts (2021).
\newblock {d-blink: Distributed End-to-End Bayesian Entity Resolution}.
\newblock {\em Journal of Computational and Graphical Statistics\/}~{\em
  30\/}(2), 406--421.

\bibitem[\protect\citeauthoryear{McVeigh, Spahn, and Murray}{McVeigh
  et~al.}{2019}]{mcveigh2020scaling}
McVeigh, B.~S., B.~T. Spahn, and J.~S. Murray (2019).
\newblock {Scaling Bayesian Probabilistic Record Linkage with Post-Hoc
  Blocking: An Application to the California Great Registers}.
\newblock \url{https://arxiv.org/abs/1905.05337}.

\bibitem[\protect\citeauthoryear{Sadinle}{Sadinle}{2014}]{sadinle2014detecting}
Sadinle, M. (2014).
\newblock {Detecting duplicates in a homicide registry using a Bayesian
  partitioning approach}.
\newblock {\em The Annals of Applied Statistics\/}~{\em 8\/}(4), 2404 -- 2434.

\bibitem[\protect\citeauthoryear{Sadinle}{Sadinle}{2017}]{sadinle2017bayesian}
Sadinle, M. (2017).
\newblock {Bayesian Estimation of Bipartite Matchings for Record Linkage}.
\newblock {\em Journal of the American Statistical Association\/}~{\em
  112\/}(518), 600--612.

\bibitem[\protect\citeauthoryear{Sadinle and Fienberg}{Sadinle and
  Fienberg}{2013}]{sadinle2013generalized}
Sadinle, M. and S.~E. Fienberg (2013).
\newblock {A Generalized Fellegi–Sunter Framework for Multiple Record Linkage
  With Application to Homicide Record Systems}.
\newblock {\em Journal of the American Statistical Association\/}~{\em
  108\/}(502), 385--397.

\bibitem[\protect\citeauthoryear{Steorts}{Steorts}{2015}]{steorts2015entity}
Steorts, R.~C. (2015).
\newblock {Entity Resolution with Empirically Motivated Priors}.
\newblock {\em Bayesian Analysis\/}~{\em 10\/}(4), 849 -- 875.

\bibitem[\protect\citeauthoryear{Steorts, Hall, and Fienberg}{Steorts
  et~al.}{2016}]{steorts2016bayesian}
Steorts, R.~C., R.~Hall, and S.~E. Fienberg (2016).
\newblock {A Bayesian Approach to Graphical Record Linkage and Deduplication}.
\newblock {\em Journal of the American Statistical Association\/}~{\em
  111\/}(516), 1660--1672.

\bibitem[\protect\citeauthoryear{Tancredi and Liseo}{Tancredi and
  Liseo}{2011}]{tancredi2011hierarchical}
Tancredi, A. and B.~Liseo (2011).
\newblock {A hierarchical Bayesian approach to record linkage and population
  size problems}.
\newblock {\em The Annals of Applied Statistics\/}~{\em 5\/}(2B), 1553 -- 1585.

\bibitem[\protect\citeauthoryear{Taylor, Kaplan, and Betancourt}{Taylor
  et~al.}{2022}]{pkgbstrl}
Taylor, I., A.~Kaplan, and B.~Betancourt (2022).
\newblock {\em bstrl: Bayesian Streaming Record Linkage}.
\newblock R package version 1.0.2.

\bibitem[\protect\citeauthoryear{Tran, Vatsalan, and Christen}{Tran
  et~al.}{2013}]{tran2013geco}
Tran, K.-N., D.~Vatsalan, and P.~Christen (2013).
\newblock {GeCo: An Online Personal Data Generator and Corruptor}.
\newblock In {\em Proceedings of the 22nd ACM International Conference on
  Information and Knowledge Management}, CIKM '13, New York, NY, USA, pp.\
  2473–2476. Association for Computing Machinery.

\bibitem[\protect\citeauthoryear{Vatsalan, Sehili, Christen, and Rahm}{Vatsalan
  et~al.}{2017}]{Vatsalan2017privacy}
Vatsalan, D., Z.~Sehili, P.~Christen, and E.~Rahm (2017).
\newblock {\em {Privacy-Preserving Record Linkage for Big Data: Current
  Approaches and Research Challenges}}, pp.\  851--895.
\newblock Cham: Springer International Publishing.

\bibitem[\protect\citeauthoryear{Winkler}{Winkler}{2006}]{winkler_2006}
Winkler, W.~E. (2006).
\newblock {Overview of Record Linkage and Current Research Directions}.
\newblock Technical report, U.S. Bureau of the Census Statistical Research
  Division.

\bibitem[\protect\citeauthoryear{Wortman}{Wortman}{2019}]{wortman2019record}
Wortman, J. P.~H. (2019).
\newblock {\em {Record Linkage Methods with Applications to Causal Inference
  and Election Voting Data}}.
\newblock Ph.\ D. thesis, Duke University.

\bibitem[\protect\citeauthoryear{Yang and Dunson}{Yang and
  Dunson}{2013}]{yang2013sequential}
Yang, Y. and D.~B. Dunson (2013).
\newblock {Sequential Markov Chain Monte Carlo}.
\newblock \url{https://arxiv.org/abs/1308.3861}.

\bibitem[\protect\citeauthoryear{Zanella}{Zanella}{2020}]{zanella2020informed}
Zanella, G. (2020).
\newblock {Informed Proposals for Local MCMC in Discrete Spaces}.
\newblock {\em Journal of the American Statistical Association\/}~{\em
  115\/}(530), 852--865.

\end{thebibliography}


\clearpage

\hypertarget{appendix-appendix}{%
\appendix}

\hypertarget{sec:supplemental-figures}{%
\section{Supplemental Figures and Tables}\label{sec:supplemental-figures}}

Figure \ref{fig:streaming-files} depicts the streaming record linkage problem up to time \(T_k\).

Table \ref{tab:f1score-supp} and Figure \ref{fig:entityerrors-supp} show F1-scores and entity errors for additional error levels from the simulation in Section \ref{sec:simulation-study}.

\begin{figure}
\centering
\includegraphics{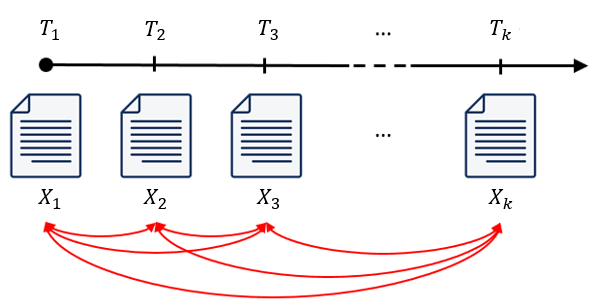}
\caption{\label{fig:streaming-files} A depiction of the streaming record linkage problem up to time \(T_k\). Files 1 through \(k\) arrive sequentially and are duplicate-free. The red arrows illustrate the growing complexity of the linkage problem on multiple files: with \(k\) files, records in \(k(k-1)/2\) pairs of files must be compared and linked.}
\end{figure}

\renewcommand{\arraystretch}{0.80}
\begin{table}

\caption{\label{tab:f1score-supp}Posterior means and standard deviations of $F_1$-score for simulated datasets. Within rows, each model is listed: the model presented in this paper (Streaming) and three comparison models. Larger values represent more accurate links in the posterior distribution. The support vector machine, a non-bayesian method, is represented only by the $F_1$-score of its resulting point estimate.}
\centering
\fontsize{10}{12}\selectfont
\begin{tabular}[t]{lllll}
\toprule
\begingroup\fontsize{10}{12}\selectfont Model\endgroup & \begingroup\fontsize{10}{12}\selectfont 10\% overlap\endgroup & \begingroup\fontsize{10}{12}\selectfont 30\% overlap\endgroup & \begingroup\fontsize{10}{12}\selectfont 50\% overlap\endgroup & \begingroup\fontsize{10}{12}\selectfont 90\% overlap\endgroup\\
\midrule
\addlinespace[0.3em]
\multicolumn{5}{l}{\textbf{Errors: 1}}\\
\hspace{1em}Streaming (Flat Prior) & 0.999 (0.0026) & 0.988 (0.0003) & \textbf{0.999 (0.0010)} & 0.998 (0.0000)\\
\hspace{1em}Streaming (Weak Prior) & 0.998 (0.0038) & \textbf{1.000 (0.0001)} & \textbf{0.999 (0.0012)} & \textbf{1.000 (0.0000)}\\
\hspace{1em}Streaming (Strong Prior) & 0.988 (0.0067) & 0.995 (0.0016) & 0.992 (0.0011) & \textbf{1.000 (0.0000)}\\
\hspace{1em}Multilink & 0.987 (0.0088) & 0.995 (0.0021) & 0.982 (0.0010) & 0.915 (0.0000)\\
\hspace{1em}Multilink (Single Likelihood) & 0.999 (0.0035) & 0.995 (0.0004) & 0.991 (0.0011) & 0.946 (0.0000)\\
\hspace{1em}Blink & 0.869 (0.0136) & 0.988 (0.0007) & \textbf{0.999 (0.0007)} & \textbf{1.000 (0.0000)}\\
\hspace{1em}SVM (1\% training) & \textbf{1.000} & 0.998 & 0.997 & \textbf{1.000}\\
\addlinespace[0.3em]
\multicolumn{5}{l}{\textbf{Errors: 3}}\\
\hspace{1em}Streaming (Flat Prior) & 0.955 (0.0195) & 0.990 (0.0058) & 0.987 (0.0021) & 0.995 (0.0001)\\
\hspace{1em}Streaming (Weak Prior) & 0.970 (0.0193) & 0.990 (0.0057) & \textbf{0.996 (0.0021)} & \textbf{1.000 (0.0001)}\\
\hspace{1em}Streaming (Strong Prior) & \textbf{0.978 (0.0159)} & 0.983 (0.0055) & \textbf{0.996 (0.0022)} & \textbf{1.000 (0.0001)}\\
\hspace{1em}Multilink & 0.095 (0.0055) & 0.981 (0.0059) & 0.983 (0.0027) & 0.954 (0.0000)\\
\hspace{1em}Multilink (Single Likelihood) & 0.940 (0.0210) & \textbf{0.991 (0.0052)} & 0.985 (0.0023) & \textbf{1.000 (0.0000)}\\
\hspace{1em}Blink & 0.543 (0.0176) & 0.944 (0.0031) & 0.988 (0.0023) & 0.999 (0.0002)\\
\hspace{1em}SVM (1\% training) & 0.933 & 0.958 & 0.984 & 0.974\\
\addlinespace[0.3em]
\multicolumn{5}{l}{\textbf{Errors: 8}}\\
\hspace{1em}Streaming (Flat Prior) & 0.231 (0.0077) & 0.414 (0.0093) & 0.822 (0.0163) & 0.950 (0.0031)\\
\hspace{1em}Streaming (Weak Prior) & 0.240 (0.0085) & 0.415 (0.0103) & 0.843 (0.0157) & 0.911 (0.0026)\\
\hspace{1em}Streaming (Strong Prior) & \textbf{0.710 (0.0277)} & \textbf{0.817 (0.0128)} & \textbf{0.898 (0.0075)} & 0.908 (0.0030)\\
\hspace{1em}Multilink & 0.204 (0.0084) & 0.372 (0.0084) & 0.647 (0.0097) & \textbf{0.977 (0.0032)}\\
\hspace{1em}Multilink (Single Likelihood) & 0.136 (0.0057) & 0.369 (0.0070) & 0.647 (0.0097) & 0.972 (0.0022)\\
\hspace{1em}Blink & 0.340 (0.0216) & 0.663 (0.0104) & 0.836 (0.0140) & 0.918 (0.0080)\\
\hspace{1em}SVM (1\% training) & 0.586 & 0.482 & 0.556 & 0.542\\
\bottomrule
\end{tabular}
\end{table}
\renewcommand{\arraystretch}{\originalarraystretch}

\begin{figure}
\centering
\includegraphics{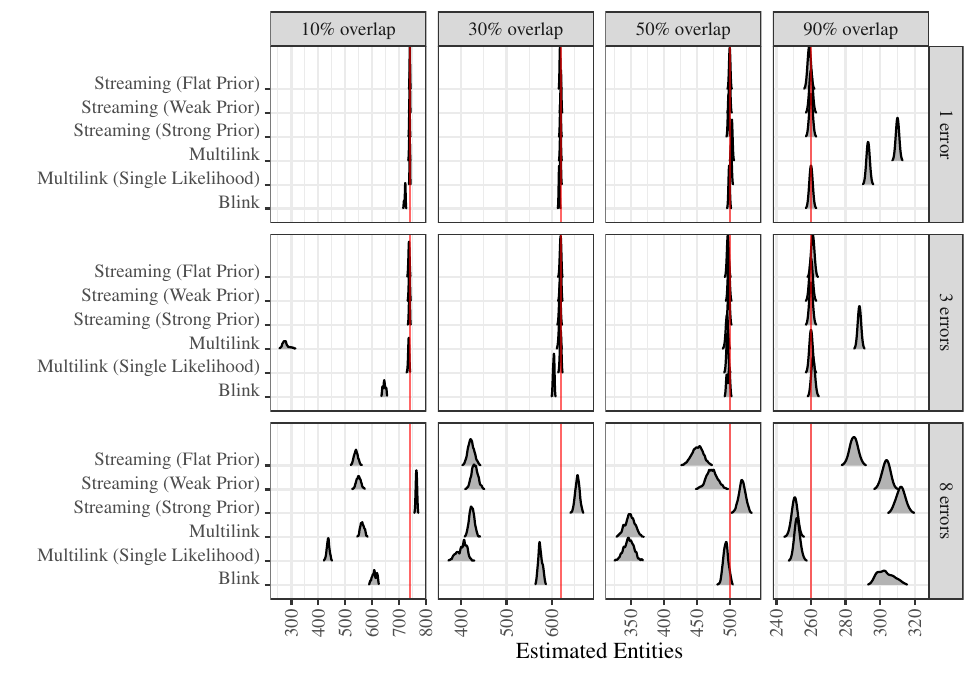}
\caption{\label{fig:entityerrors-supp}Posterior distribution of the number of estimated entities for simulated datasets. A vertical line indicates the true number of distinct entities in each dataset. Distributions to the right or left of the vertical line indicate underlinking or overlinking, respectively, in the posterior. Compared models are on the y-axis: the model presented in this paper (Streaming) and three comparison models.}
\end{figure}

\hypertarget{posterior-and-full-conditional-distributions}{%
\section{Posterior and Full Conditional Distributions}\label{posterior-and-full-conditional-distributions}}

\hypertarget{sec:posterior-distribution}{%
\subsection{Posterior Distribution}\label{sec:posterior-distribution}}

Here we specify a function that is proportional to the full streaming record linkage posterior density.

\begin{align}
&P(\boldsymbol m, \boldsymbol u, \boldsymbol Z^{(1)}, \dots, \boldsymbol Z^{(k-1)} | \Gamma^{(1)}, \dots, \Gamma^{(k-1)}) \\
&\quad\propto P(\boldsymbol m) P(\boldsymbol u) P(\boldsymbol Z^{(1)}) \cdots P(\boldsymbol Z^{(k-1)}) P(\Gamma^{(1)}, \dots, \Gamma^{(k-1)} | \boldsymbol m, \boldsymbol u, \boldsymbol Z^{(1)}, \dots, \boldsymbol Z^{(k-1)}) \\
&\quad\propto \prod_{f=1}^F\prod_{\ell=0}^{L_f}m_{f\ell}^{a_{f\ell}}u_{f\ell}^{b_{f\ell}} \nonumber \\
&\quad\quad\times \prod_{t=2}^{k} \left[ \frac{(N_{t-1} - n_{t\cdot}(\boldsymbol Z^{(t-1)}))!}{N_{t-1}!}\cdot\frac{\mbox{B}(n_{t\cdot}(\boldsymbol Z^{(t-1)}) + \alpha_\pi, n_t - n_{t\cdot}(\boldsymbol Z^{(t-1)}) + \beta_\pi)}{\mbox{B}(\alpha_\pi, \beta_\pi)} \right] \nonumber \\
&\quad\quad\times \prod_{t_1 < t_2}^k \prod_{i=1}^{n_{t_1}} \prod_{j=1}^{n_{t_2}} \prod_{f=1}^F \prod_{\ell=0}^{L_f} \left[m_{f\ell}^{\mathbb{I}((\boldsymbol x_{t_1 i}, \boldsymbol x_{t_2 j}) \in M)} u_{f\ell}^{\mathbb{I}((\boldsymbol x_{t_1 i}, \boldsymbol x_{t_2 j}) \notin M)} \right]^{\gamma^{f\ell}(\boldsymbol x_{t_1 i}, \boldsymbol x_{t_2 j})}, \label{eqn:posterior-distribution}
\end{align}
where
\begin{align*}
N_{t-1} &= n_1 + \dots + n_{t-1} \\
n_{t\cdot}(\boldsymbol Z^{(t-1)}) &= \sum_{j=1}^{n_t} \mathbb{I}(Z^{(t-1)}_j \leq N_{t-1}) \\
M &:= M(\boldsymbol Z^{(1)}, \dots, \boldsymbol Z^{(k-1)}) = \{(\boldsymbol x_{t_1 i}, \boldsymbol x_{t_2 j}): \boldsymbol x_{t_1 i}\text{ and }\boldsymbol x_{t_2 j}\text{ are linked}\}.
\end{align*}

\hypertarget{sec:full-conditional-distributions-m-u}{%
\subsection{\texorpdfstring{Full conditional for \(\boldsymbol m\), \(\boldsymbol u\)}{Full conditional for \textbackslash boldsymbol m, \textbackslash boldsymbol u}}\label{sec:full-conditional-distributions-m-u}}

We provide the full conditional distribution for \(\boldsymbol m\) starting from the posterior in Equation \ref{eqn:posterior-distribution}.

\begin{align}
&P(\boldsymbol m | \boldsymbol u, \boldsymbol Z^{(1)}, \dots, \boldsymbol Z^{(k-1)}, \Gamma^{(1)}, \dots, \Gamma^{(k-1)}) \\
&\quad\propto P(\boldsymbol m, \boldsymbol u, \boldsymbol Z^{(1)}, \dots, \boldsymbol Z^{(k-1)} | \Gamma^{(1)}, \dots, \Gamma^{(k-1)}) \\
&\quad\propto \prod_{f=1}^F\prod_{\ell=0}^{L_f} m_{f\ell}^{ a_{f\ell} + \sum_{t_1 < t_2}^k \sum_{i=1}^{n_{t_1}} \sum_{j=1}^{n_{t_2}} \mathbb{I}((\boldsymbol x_{t_1 i}, \boldsymbol x_{t_2 j}) \in M)\cdot\gamma^{f\ell}(\boldsymbol x_{t_1 i}, \boldsymbol x_{t_2 j})}. \label{eqn:m-full-conditional}
\end{align}

We recognize the inside products in Equation \ref{eqn:m-full-conditional} as the kernel of a Dirichlet distribution, and so each vector \(\boldsymbol m_f\) for \(f=1,\dots,F\) has a conjugate Dirichlet full conditional distribution. Similarly, we can derive

\begin{equation}
P(\boldsymbol u | \boldsymbol m, \boldsymbol Z^{(1)}, \dots, \boldsymbol Z^{(k-1)}, \Gamma^{(1)}, \dots, \Gamma^{(k-1)})
\propto
\prod_{f=1}^F\prod_{\ell=0}^{L_f} u_{f\ell}^{ b_{f\ell} + \sum_{t_1 < t_2}^k \sum_{i=1}^{n_{t_1}} \sum_{j=1}^{n_{t_2}} \mathbb{I}((\boldsymbol x_{t_1 i}, \boldsymbol x_{t_2 j}) \notin M)\cdot\gamma^{f\ell}(\boldsymbol x_{t_1 i}, \boldsymbol x_{t_2 j})},
\label{eqn:u-full-conditional}
\end{equation}
and so each vector \(\boldsymbol u_f\) for \(f=1,\dots,F\) also has a conjugate Dirichlet full conditional distribution.

\hypertarget{sec:full-conditional-distributions-z}{%
\subsection{\texorpdfstring{Full conditional for \(\boldsymbol Z^{(t-1)}\)}{Full conditional for \textbackslash boldsymbol Z\^{}\{(t-1)\}}}\label{sec:full-conditional-distributions-z}}

Let \(T\) be a file number, \(2 \leq T \leq k\). We derive the full conditional distribution for \(\boldsymbol Z^{(T-1)}\), the matching vector introduced with file \(X_T\), starting from the posterior in Equation \ref{eqn:posterior-distribution}.

\begin{align}
&P(\boldsymbol Z^{(T-1)} | \boldsymbol m, \boldsymbol u, \boldsymbol Z^{(1)}, \dots,\boldsymbol Z^{(T-2)}, \boldsymbol Z^{(T)}, \dots, \boldsymbol Z^{(k-1)}, \Gamma^{(1)}, \dots, \Gamma^{(k-1)}) \\
&\quad\propto P(\boldsymbol m, \boldsymbol u, \boldsymbol Z^{(1)}, \dots, \boldsymbol Z^{(k-1)} | \Gamma^{(1)}, \dots, \Gamma^{(k-1)}) \\
&\quad\propto \left[ \frac{(N_{T-1} - n_{T\cdot}(\boldsymbol Z^{(T-1)}))!}{N_{T-1}!}\cdot\frac{\mbox{B}(n_{T\cdot}(\boldsymbol Z^{(T-1)}) + \alpha_\pi, n_T - n_{T\cdot}(\boldsymbol Z^{(T-1)}) + \beta_\pi)}{\mbox{B}(\alpha_\pi, \beta_\pi)} \right] \nonumber \\
&\quad\quad\times \prod_{t_2=T}^k \prod_{t_1=1}^{t_2-1} \prod_{i=1}^{n_{t_1}} \prod_{j=1}^{n_{t_2}} \prod_{f=1}^F \prod_{\ell=0}^{L_f} \left[m_{f\ell}^{\mathbb{I}((\boldsymbol x_{t_1 i}, \boldsymbol x_{t_2 j}) \in M)} u_{f\ell}^{\mathbb{I}((\boldsymbol x_{t_1 i}, \boldsymbol x_{t_2 j}) \notin M)} \right]^{\gamma^{f\ell}(\boldsymbol x_{t_1 i}, \boldsymbol x_{t_2 j})}, \label{eqn:z-full-conditional}
\end{align}
where
\begin{align*}
N_{t-1} &= n_1 + \dots + n_{t-1} \\
n_{t\cdot}(\boldsymbol Z^{(t-1)}) &= \sum_{j=1}^{n_t} \mathbb{I}(Z^{(t-1)}_j \leq N_{t-1}) \\
M &:= M(\boldsymbol Z^{(1)}, \dots, \boldsymbol Z^{(k-1)}) \\
&= \{(\boldsymbol x_{t_1 i}, \boldsymbol x_{t_2 j}): \boldsymbol x_{t_1 i}\text{ and }\boldsymbol x_{t_2 j}\text{ are linked}\}.
\end{align*}

Pairs of records, \(\boldsymbol x_{t_1 i}\) and \(\boldsymbol x_{t_2 j}\), where \(t_1, t_2 < T\) do not depend on \(\boldsymbol Z^{(T-1)}\) to be linked because of the constraints outlined in Section \ref{sec:preserving-transitivity}.

\hypertarget{supplemental-definitions-and-theorems}{%
\section{Supplemental Definitions and Theorems}\label{supplemental-definitions-and-theorems}}

\hypertarget{sec:sampler-definitions-theorems}{%
\subsection{Sampler Definitions and Theorems}\label{sec:sampler-definitions-theorems}}

\begin{definition}
\protect\hypertarget{def:gibbs-sampler-component}{}\label{def:gibbs-sampler-component}

\textbf{Component-wise sampler}. Define the component-wise sampler for sampling from the streaming record linkage model as follows:

\begin{enumerate}
\def\labelenumi{\arabic{enumi}.}
\tightlist
\item
  For \(f = 1, \dots, F\)

  \begin{enumerate}
  \def\labelenumii{\alph{enumii}.}
  \tightlist
  \item
    Update the vector \(\boldsymbol m_f\) from its conjugate full conditional Dirichlet distribution.
  \item
    Update the vector \(\boldsymbol u_f\) from its conjugate full conditional Dirichlet distribution.
  \end{enumerate}
\item
  For each vector \(\boldsymbol Z^{(\ell)}\), \(\ell = 1, \dots, k-1\)

  \begin{enumerate}
  \def\labelenumii{\alph{enumii}.}
  \tightlist
  \item
    For each index \(j = 1, \dots, n_{\ell+1}\), update the component \(Z^{(\ell)}_j\) from its full conditional distribution over all possible values, \(1, \dots, (n_1 + \dots + n_\ell), (n_1 + \dots + n_\ell + j)\).
  \end{enumerate}
\item
  Repeat steps 1 and 2 for \(s = 1, \dots, S\) times.
\end{enumerate}

\end{definition}

\begin{definition}
\protect\hypertarget{def:gibbs-sampler-lb}{}\label{def:gibbs-sampler-lb}

\textbf{Locally balanced sampler}. Define the locally balanced sampler for sampling from the streaming record linkage model as follows:

\begin{enumerate}
\def\labelenumi{\arabic{enumi}.}
\tightlist
\item
  For \(f = 1, \dots, F\)

  \begin{enumerate}
  \def\labelenumii{\alph{enumii}.}
  \tightlist
  \item
    Update the vector \(\boldsymbol m_f\) from its conjugate full conditional Dirichlet distribution.
  \item
    Update the vector \(\boldsymbol u_f\) from its conjugate full conditional Dirichlet distribution.
  \end{enumerate}
\item
  For each vector \(\boldsymbol Z^{(\ell)}\), \(\ell = 1, \dots, k-1\)

  \begin{enumerate}
  \def\labelenumii{\alph{enumii}.}
  \tightlist
  \item
    Propose a new value of \(\boldsymbol Z^{(\ell)}\) using locally balanced proposals \citep{zanella2020informed}. Each potential proposal takes a step through either the addition of a link, the removal of a link, swapping one end of a link, or exchanging ends of two links (double-swap). Proposal probabilities are weighted with barker weights, \(g(t) = t/(1+t)\).
  \item
    Accept or reject the proposal using the standard Metropolis-Hastings acceptance ratio for asymmetric proposals.
  \end{enumerate}
\item
  Repeat steps 1 and 2 for \(s = 1, \dots, S\) times.
\end{enumerate}

\end{definition}

\begin{theorem}
\protect\hypertarget{thm:component-sampling}{}\label{thm:component-sampling}The component-wise sampler (Definition \ref{def:gibbs-sampler-component}) produces an ergodic Markov chain with the streaming record linkage model posterior distribution as its target distribution.
\end{theorem}

\begin{proof}
The sampler in Definition \ref{def:gibbs-sampler-component} is a Gibbs algorithm which samples directly from the full conditional distributions of the parameters in sequence. Therefore if we prove that the resulting Markov chain is irreducible, then it is ergodic and samples from the posterior distribution. From an initial state with non-zero probability, \((\boldsymbol m, \boldsymbol u, \boldsymbol Z^{(1)}, \dots, \boldsymbol Z^{(k-1)})\), a new state with non-zero probability, \((\boldsymbol m_\ast, \boldsymbol u_\ast, \boldsymbol Z^{(1)}_\ast, \dots, \boldsymbol Z^{(k-1)}_\ast)\), may always be reached through a sequence of non-zero probability steps. For the matching vectors, first remove all existing links from \((\boldsymbol Z^{(1)}, \dots, \boldsymbol Z^{(k-1)})\) one component at a time until the completely unlinked state is reached. In the next iteration, add all links in \((\boldsymbol Z^{(1)}_\ast, \dots, \boldsymbol Z^{(k-1)}_\ast)\) one component at a time. All components of \(\boldsymbol m\) and \(\boldsymbol u\) are strictly positive, so states have zero posterior probability if and only if the state is invalid (Definition \ref{def:link-validity}) and the indicator in the likelihood equals zero. As states are invalid due to conflicting links, removing links can never turn a valid state to invalid. Since \((\boldsymbol Z^{(1)}_\ast, \dots, \boldsymbol Z^{(k-1)}_\ast)\) is valid and has nonzero posterior probability, constructing it one link at a time will never result in an invalid state.
\end{proof}

\begin{theorem}
\protect\hypertarget{thm:lb-sampling}{}\label{thm:lb-sampling}The locally balanced sampler (Definition \ref{def:gibbs-sampler-lb}) produces an ergodic Markov chain with the streaming record linkage model posterior distribution as its target distribution.
\end{theorem}

\begin{proof}
The sampler in Definition \ref{def:gibbs-sampler-lb} is a Metropolis-Hastings within Gibbs algorithm. Therefore it is sufficient to show that the resulting chain is irreducible. Similarly to the proof of Theorem \ref{thm:component-sampling}, we show there is a non-zero probability path between a starting state, \((\boldsymbol m, \boldsymbol u, \boldsymbol Z^{(1)}, \dots, \boldsymbol Z^{(k-1)})\), and an ending state, \((\boldsymbol m_\ast, \boldsymbol u_\ast, \boldsymbol Z^{(1)}_\ast, \dots, \boldsymbol Z^{(k-1)}_\ast)\), via the completely unlinked state. In each iteration, the locally balanced proposals may remove a single link or add a single link to each vector \(\boldsymbol Z^{(1)}, \dots, \boldsymbol Z^{(k-1)}\). As in the proof of Theorem \ref{thm:component-sampling}, each of these steps are to states with positive probability. Since the locally balanced proposals are weighted by the target density, they can be proposed with positive probability.
\end{proof}

\begin{theorem}
\protect\hypertarget{thm:pprb-within-gibbs-sampling-restated}{}\label{thm:pprb-within-gibbs-sampling-restated}The PPRB-within-Gibbs sampler (Definition \ref{def:pprb-within-gibbs}) produces an ergodic Markov chain with the model's posterior distribution as its target distribution if the target distribution satisfies the positivity condition,
\[p(\boldsymbol \theta_1 | \boldsymbol y_1, \boldsymbol y_2) > 0,\ p(\boldsymbol \theta_2 | \boldsymbol y_1, \boldsymbol y_2) > 0,\ p(\boldsymbol \theta_3 | \boldsymbol y_1, \boldsymbol y_2) > 0 \implies p(\boldsymbol \theta_1, \boldsymbol \theta_2, \boldsymbol \theta_3 | \boldsymbol y_1, \boldsymbol y_2) > 0.\]
\end{theorem}

\begin{proof}
First, we show that the Metropolis-Hastings acceptance ratio, \(\alpha\), in step 2 is appropriate for the target distribution. Since the proposals come from the distribution, \(p(\boldsymbol \theta_1^\ast | \boldsymbol y_1)\), the acceptance ratio would be
\begin{align*}
\alpha &= \frac{p(\boldsymbol \theta_1^\ast | \boldsymbol \theta_2, \boldsymbol \theta_3, \boldsymbol y_1, \boldsymbol y_2)}{p(\boldsymbol \theta_1 | \boldsymbol \theta_2, \boldsymbol \theta_3, \boldsymbol y_1, \boldsymbol y_2)}\frac{p(\boldsymbol \theta_1 | \boldsymbol y_1)}{p(\boldsymbol \theta_1^\ast | \boldsymbol y_1)} \\
&= \frac{p(\boldsymbol \theta_1^\ast, \boldsymbol \theta_2, \boldsymbol \theta_3 | \boldsymbol y_1, \boldsymbol y_2)}{p(\boldsymbol \theta_1, \boldsymbol \theta_2, \boldsymbol \theta_3 | \boldsymbol y_1, \boldsymbol y_2)}\frac{p(\boldsymbol \theta_1 | \boldsymbol y_1)}{p(\boldsymbol \theta_1^\ast | \boldsymbol y_1)} \\
&= \frac{p(\boldsymbol y_1 | \boldsymbol \theta_1^\ast, \boldsymbol \theta_2)p(\boldsymbol y_2 | \boldsymbol \theta_1^\ast, \boldsymbol \theta_2, \boldsymbol \theta_3)p(\boldsymbol \theta_1^\ast)p(\boldsymbol \theta_2)p(\boldsymbol \theta_3)}{p(\boldsymbol y_1 | \boldsymbol \theta_1, \boldsymbol \theta_2)p(\boldsymbol y_2 | \boldsymbol \theta_1, \boldsymbol \theta_2, \boldsymbol \theta_3)p(\boldsymbol \theta_1)p(\boldsymbol \theta_2)p(\boldsymbol \theta_3)}\frac{p(\boldsymbol \theta_1 | \boldsymbol y_1)}{p(\boldsymbol \theta_1^\ast | \boldsymbol y_1)} \\
&= \frac{p(\boldsymbol y_2 | \boldsymbol \theta_1^\ast, \boldsymbol \theta_2, \boldsymbol \theta_3)}{p(\boldsymbol y_2 | \boldsymbol \theta_1, \boldsymbol \theta_2, \boldsymbol \theta_3)}\frac{p(\boldsymbol \theta_1^\ast, \boldsymbol \theta_2 | \boldsymbol y_1)}{p(\boldsymbol \theta_1, \boldsymbol \theta_2 | \boldsymbol y_1)}\frac{p(\boldsymbol \theta_1 | \boldsymbol y_1)}{p(\boldsymbol \theta_1^\ast | \boldsymbol y_1)} \\
&= \frac{p(\boldsymbol y_2 | \boldsymbol \theta_1^\ast, \boldsymbol \theta_2, \boldsymbol \theta_3)}{p(\boldsymbol y_2 | \boldsymbol \theta_1, \boldsymbol \theta_2, \boldsymbol \theta_3)}\frac{p(\boldsymbol \theta_2 | \boldsymbol \theta_1^\ast, \boldsymbol y_1)}{p(\boldsymbol \theta_2 | \boldsymbol \theta_1, \boldsymbol y_1)}.
\end{align*}

Second, we have that \(p(\boldsymbol \theta_1 | \boldsymbol y_1) = 0 \implies p(\boldsymbol \theta_1 | \boldsymbol \theta_2, \boldsymbol \theta_3, \boldsymbol y_1, \boldsymbol y_2) = 0\) since the latter distribution is conditioned on a superset of random variables as the former. Therefore the distribution \(p(\boldsymbol \theta_1 | \boldsymbol y_1)\) works as an independent Metropolis-Hastings proposal distribution for the target \(p(\boldsymbol \theta_1 | \boldsymbol \theta_2, \boldsymbol \theta_3, \boldsymbol y_1, \boldsymbol y_2)\).

Finally, the positivity condition implies that a Gibbs sampler is irreducible, and so the algorithm produces an ergodic Markov chain. \citep{robert2005monte}
\end{proof}

\hypertarget{sec:pprb-within-gibbs-streamingrl}{%
\subsubsection{PPRB-within-Gibbs sampler for Streaming RL model}\label{sec:pprb-within-gibbs-streamingrl}}

We perform the three steps of each iteration as

\begin{enumerate}
\def\labelenumi{\arabic{enumi}.}
\tightlist
\item
  For \(f = 1, \dots, F\)

  \begin{enumerate}
  \def\labelenumii{\alph{enumii}.}
  \tightlist
  \item
    Update the vector \(\boldsymbol m_f\) from its conjugate full conditional Dirichlet distribution (see Appendix \ref{sec:full-conditional-distributions-m-u}).
  \item
    Update the vector \(\boldsymbol u_f\) from its conjugate full conditional Dirichlet distribution (see Appendix \ref{sec:full-conditional-distributions-m-u}).
  \end{enumerate}
\item
  (PPRB step) Propose a new value \((\boldsymbol Z^{(1)}_\ast, \dots, \boldsymbol Z^{(k-2)}_\ast)\) by drawing from the existing posterior samples (with replacement). Accept or reject the proposal using the Metropolis-Hastings ratio, \[\alpha = \min\left(\frac{p(\Gamma^{(k-1)} | \boldsymbol Z^{(1)}_\ast,\dots, \boldsymbol Z^{(k-2)}_\ast, \boldsymbol m, \boldsymbol u, \boldsymbol Z^{(k-1)})}{p(\Gamma^{(k-1)} | \boldsymbol Z^{(1)},\dots, \boldsymbol Z^{(k-2)}, \boldsymbol m, \boldsymbol u, \boldsymbol Z^{(k-1)})} \cdot \frac{p(\boldsymbol m, \boldsymbol u | \boldsymbol Z^{(1)}_\ast,\dots, \boldsymbol Z^{(k-2)}_\ast, \Gamma^{(1)}, \dots, \Gamma^{(k-2)})}{p(\boldsymbol m, \boldsymbol u | \boldsymbol Z^{(1)},\dots, \boldsymbol Z^{(k-2)}, \Gamma^{(1)}, \dots, \Gamma^{(k-2)})}, 1\right)\]
\item
  Update the value of \(\boldsymbol Z^{(k-1)}\) using a Metropolis-Hastings proposal targeting its full conditional distribution. We examine two such possible proposals in Section \ref{sec:z-proposals}.
\end{enumerate}

\hypertarget{simulation-and-sampling-details}{%
\section{Simulation and Sampling Details}\label{simulation-and-sampling-details}}

This appendix contains details for MCMC runs and simulation studies whose results are presented in the main body of the paper.

\hypertarget{sec:link-accuracy-model-details}{%
\subsection{Link Accuracy Comparison}\label{sec:link-accuracy-model-details}}

\noindent \textbf{Proposed Model: Streaming Record Linkage}

The sampler was run for 2500 iterations, discarding the first 500. We set \(\alpha_\pi = \beta_\pi = 1\) as an uninformative prior for \(\boldsymbol Z^{(1)}\), \(\boldsymbol Z^{(2)}\), and \(\boldsymbol Z^{(3)}\). Flat Dirichlet priors were chosen for \(\boldsymbol u\), and three choices of prior strength were used for \(\boldsymbol m\) (Flat, Weak, Strong). Component-wise proposals were used for \(\boldsymbol Z^{(1)}\), \(\boldsymbol Z^{(2)}\), and \(\boldsymbol Z^{(3)}\) to avoid needing excessive burn-in. We found that a Gibbs sampler with locally balanced proposals required too many iterations to converge to the target posterior distribution to be computationally feasible.

\noindent \textbf{Multilink} \citep{aleshinguendel2021multifile}

We use flat Dirichlet priors for the \(\boldsymbol m\) and \(\boldsymbol u\) parameters, \(\boldsymbol \alpha = \boldsymbol 1\) for the Dirichlet-multinomial overlap table prior on the partitions and a uniform prior on the number of clusters. For each of the simulated datasets, we produce 1000 posterior samples after a 500 iteration burn-in from an initial state of no linked pairs.

\noindent \textbf{Blink} \citep{steorts2015entity}

For string fields, we choose a steepness parameter \(c=1\) and the generalized Levenshtein distance of the \texttt{R} function \texttt{adist}. For categorical fields, we choose beta parameters \(a=5\) and \(b=20\) to encode prior knowledge of between 1 and 4 errors per record, or a distortion probability of between 0.1 and 0.4. For each simulated dataset, we produce 1000 posterior samples after a 5000 iteration burn-in.

\noindent \textbf{Support Vector Machine}

Training pairs were chosen as evenly as possible between coreferent and non-coreferent pairs, which sometimes resulted in all coreferent pairs being included in the training set.

\hypertarget{sec:speed-comparison-sampler-details}{%
\subsection{Speed Comparison}\label{sec:speed-comparison-sampler-details}}

The Gibbs sampler was run using component-wise full conditional updates for \(\boldsymbol Z^{(1)}\), \(\boldsymbol Z^{(2)}\) and \(\boldsymbol Z^{(3)}\) for 2500 iterations, discarding the first 500 for burn-in. Each PPRB update was run for 5000 iterations, discarding the first 1000 for burn-in. The SMCMC updates used ensembles of size 200 and were computed with 12 parallel processes. SMCMC-Comp used 5 jumping kernel iterations and 50 transition kernel iterations, SMCMC-LB used 50 jumping kernel iterations and and 200 transition kernel iterations, and SMCMC-Mixed used 5 jumping kernel iterations and 200 transition kernel iterations. All locally balanced proposals used a block size of 75 records.

\hypertarget{sec:poland-mcmc-details}{%
\subsection{Social Diagnosis Survey Analysis}\label{sec:poland-mcmc-details}}

The Gibbs sampler was run for 2500 iterations, discarding the first 500 for burn-in. Each PPRB update was run for 5000 iterations, discarding the first 1000 for burn-in. The SMCMC updates used ensembles of size 200 and were computed with 12 parallel processes. SMCMC-Comp used 5 jumping kernel iterations and 50 transition kernel iterations, SMCMC-LB used 500 jumping kernel iterations and and 200 transition kernel iterations, and SMCMC-Mixed used 5 jumping kernel iterations and 200 transition kernel iterations. All locally balanced proposals used a block size of 150 records.

\end{document}